\documentclass[12pt,final]{article}
\usepackage{mathpazo}
\usepackage{amssymb,amsthm,amsmath,amsfonts,enumitem,bbm,xr-hyper}
\usepackage{xcolor}
\usepackage{amsmath}
\usepackage{amsthm}
\makeatletter
\def\thm@space@setup{%
  \thm@preskip=\parskip \thm@postskip=0pt
}
\makeatother
\usepackage{setspace}
\usepackage{fnpct}
\usepackage{amsfonts}
\usepackage{graphicx}
\usepackage{subfig}
\usepackage{natbib}
\usepackage[colorlinks,linkcolor=links,citecolor=cites,urlcolor=MyDarkBlue]{hyperref}
\usepackage{geometry}
\usepackage{parskip}
\usepackage[justification=centering]{caption}
\newcommand{\subscript}[2]{$#1 _ #2$}
\usepackage{calrsfs}
\DeclareMathAlphabet{\pazocal}{OMS}{zplm}{m}{n}
\definecolor{MyDarkBlue}{rgb}{0,0.08,0.45}
\definecolor{cites}{HTML}{324b13}
\definecolor{links}{HTML}{1a663b}
\definecolor{MyLightMagenta}{cmyk}{0.1,0.8,0,0.1}
\theoremstyle{definition}\newtheorem{definition}{Definition}
\theoremstyle{definition}\newtheorem{rem}{Remark}
\theoremstyle{definition}\newtheorem{cor}{Corollary}
\theoremstyle{theorem}\newtheorem{theorem}{Theorem}
\theoremstyle{theorem}\newtheorem{prop}{Proposition}
\theoremstyle{theorem}\newtheorem{lemma}{Lemma}
\theoremstyle{definition}\newtheorem{example}{Example}
\theoremstyle{definition}\newtheorem{ex}{Example}

\newcommand{\genericu}{\ensuremath{u(a,b)}}
\newcommand{\genericuo}{\ensuremath{u(a,}}
\newcommand{\genericv}{\ensuremath{v(a,b)}}

\newcommand{\cala}{\ensuremath{\pazocal{A}}}
\newcommand{\calb}{\ensuremath{\pazocal{B}}}
\newcommand{\mada}{\ensuremath{m^{A-DA}}}

\newcommand{\mbda}{\ensuremath{m^{B-DA}}}

\newcommand{\sidea}{\ensuremath{\mathcal{A}}}
\newcommand{\sideb}{\ensuremath{\mathcal{B}}}

\newcommand{\arrivalaupto}{\ensuremath{\overline{A}}}
\newcommand{\arrivalbupto}{\ensuremath{\overline{B}}}
\newcommand{\arrivalauptot}{\ensuremath{\arrivalaupto_t}}
\newcommand{\arrivalbuptot}{\ensuremath{\arrivalbupto_t}}

\newcommand{\arrivalauptoterminal}{\ensuremath{\arrivalaupto_T}}
\newcommand{\arrivalbuptoterminal}{\ensuremath{\arrivalbupto_T}}
\newcommand{\matching}{\ensuremath{m}}
\newcommand{\matchingb}{\ensuremath{\overline{\matching}}}
\newcommand{\matchingbt}{\ensuremath{\matchingb_t}}

\newcommand{\matchingstar}{\ensuremath{\matching^\star}}

\newcommand{\matchingstaruptoone}{\ensuremath{\matching^{\star^{1}}}}
\newcommand{\matchingstaruptot}{\ensuremath{\matching^{\star^{t-1}}}}
\newcommand{\matchingstaruptoterminal}{\ensuremath{\matching^{\star^{\terminal-1}}}}
\newcommand{\matchingt}{\ensuremath{\matching_t}}

\newcommand{\matchinguptot}{\ensuremath{\matching^{t-1}}}

\newcommand{\terminal}{\ensuremath{T}}

\newcommand{\Matchings}{\ensuremath{\pazocal{M}}}
\newcommand{\Matchingsfinal}{\ensuremath{\Matchings_\terminal}}

\newcommand{\ds}{\ensuremath{\pazocal{D}}}

\newcommand{\Gterminal}{\ensuremath{G_\terminal}}

\newcommand{\Gt}{\ensuremath{G_t}}
\newcommand{\Gtwo}{\ensuremath{G_2}}
\newcommand{\Gone}{\ensuremath{G_1}}
\newcommand{\periodone}{\ensuremath{t=1}}
\newcommand{\periodtwo}{\ensuremath{t=2}}

\newcommand{\typea}{\ensuremath{\mathbf{a}}}
\newcommand{\typeauptot}{\ensuremath{\overline{\typea}_t}}
\newcommand{\typeauptoterminal}{\ensuremath{\overline{\typea}_\terminal}}

\newcommand{\typeat}{\ensuremath{\typea_t}}
\newcommand{\typeas}{\ensuremath{\typea_\dateindex}}
\newcommand{\typeb}{\ensuremath{\mathbf{b}}}
\newcommand{\typebuptot}{\ensuremath{\overline{\typeb}_t}}

\newcommand{\typebuptoterminal}{\ensuremath{\overline{\typeb}_\terminal}}
\newcommand{\typebt}{\ensuremath{\typeb_t}}
\newcommand{\typebs}{\ensuremath{\typeb_\dateindex}}
\newcommand{\supportterminal}{\ensuremath{\overline{\economy}^\terminal}}
\newcommand{\supportt}{\ensuremath{\overline{\economy}^t}}

\newcommand{\economylmps}{\ensuremath{(I,\typea,J,\typeb)}}
\newcommand{\economytlmps}{\ensuremath{(I_t,\typeat,J_t,\typebt)}}
\newcommand{\economyslmps}{\ensuremath{(I_\dateindex,\typeas,J_\dateindex,\typebs)}}
\newcommand{\arrivaliupto}{\ensuremath{\overline{I}}}
\newcommand{\arrivaljupto}{\ensuremath{\overline{J}}}
\newcommand{\arrivaliuptot}{\ensuremath{\arrivaliupto_t}}

\newcommand{\arrivaliuptos}{\ensuremath{\arrivaliupto_\dateindex}}
\newcommand{\arrivaljuptos}{\ensuremath{\arrivaljupto_\dateindex}}

\newcommand{\cali}{\ensuremath{\pazocal{I}}}
\newcommand{\calj}{\ensuremath{\pazocal{J}}}

\newcommand{\dateindex}{\ensuremath{s}}

\newcommand{\agent}{\ensuremath{k}}
\newcommand{\agentp}{\ensuremath{\agent^\prime}}
\newcommand{\agenta}{\ensuremath{a}}
\newcommand{\agentap}{\ensuremath{\agenta^\prime}}

\newcommand{\agentb}{\ensuremath{b}}
\newcommand{\agentbp}{\ensuremath{\agentb^\prime}}

\newcommand{\economy}{\ensuremath{\mathcal{E}}}
\newcommand{\economyterminal}{\ensuremath{\economy^\terminal}}
\newcommand{\economyt}{\ensuremath{\economy^t}}
\newcommand{\economyr}{\ensuremath{E}}
\newcommand{\economyrtwo}{\ensuremath{\economyr^2}}
\newcommand{\economyrone}{\ensuremath{\economyr^1}}
\newcommand{\economyrterminal}{\ensuremath{\economyr^\terminal}}
\newcommand{\economyrt}{\ensuremath{\economyr^t}}
\newcommand{\economyruptot}{\ensuremath{\economyr^{t-1}}}
\newcommand{\economyrs}{\ensuremath{\economyr^\dateindex}}
\newcommand{\arrivala}{\ensuremath{A}}

\newcommand{\arrivalb}{\ensuremath{B}}

\newcommand{\agenti}{\ensuremath{i}}
\newcommand{\agentj}{\ensuremath{j}}
\newcommand{\arrivali}{\ensuremath{I}}

\newcommand{\arrivalj}{\ensuremath{J}}

\newcommand{\agentip}{\ensuremath{\agenti^\prime}}
\newcommand{\agentjp}{\ensuremath{\agentj^\prime}}
\newcommand{\genericuione}{\ensuremath{u(\typea_1(i),\typeb_1(j))}}
\newcommand{\genericuit}{\ensuremath{u(\typeauptot(i),\typebuptot(j))}}
\newcommand{\genericvjone}{\ensuremath{v(\typea_1(i),\typeb_1(j))}}
\newcommand{\genericvjt}{\ensuremath{v(\typeauptot(i),\typebuptot(j))}}

\newcommand{\erdos}{\ensuremath{\text{Erd\H{o}s}}}
\newcommand{\kuhn}{\ensuremath{\text{Kuhn}}}
\newcommand{\gale}{\ensuremath{\text{Gale}}}

\newcommand{\renyi}{\ensuremath{\text{R\'enyi}}}
\newcommand{\tucker}{\ensuremath{\text{Shapley}}}
\newcommand{\shapley}{\ensuremath{\text{Tucker}}}
\newcommand{\nash}{\ensuremath{\text{Nash}}}

\newcommand{\spot}{\ensuremath{\nu}}

\newcommand{\stable}{\ensuremath{\pazocal{S}}}

\newcommand{\economyrb}{\ensuremath{\hat{\economyr}}}

\newcommand{\economyrbt}{\ensuremath{\economyrb^t}}

\newcommand{\arrivalab}{\ensuremath{\hat{\arrivala}}}
\newcommand{\arrivalbb}{\ensuremath{\hat{\arrivalb}}}

\newcommand{\matchingd}{\ensuremath{\hat{\matching}}}
\newcommand{\matchingduptot}{\ensuremath{\matchingd^{t-1}}}
\newcommand{\strat}{\ensuremath{\sigma}}
\newcommand{\stratb}{\ensuremath{\strat^\prime}}
\newcommand{\publict}{\ensuremath{h^t}}

\newcommand{\pair}{\ensuremath{(\agenta,\agentb)}}
\newcommand{\rol}{\ensuremath{\succ}}
\newcommand{\rolb}{\ensuremath{\rol^\prime}}
\newcommand{\profilerol}{\ensuremath{\overline{\rol}}}
\newcommand{\gameda}{\ensuremath{\Gamma_1}}
\newcommand{\gamestable}{\ensuremath{\Gamma_2}}

\makeatletter
\newcommand*{\addFileDependency}[1]{
  \typeout{(#1)}
  \@addtofilelist{#1}
  \IfFileExists{#1}{}{\typeout{No file #1.}}
}
\makeatother

\title{Dynamically Stable Matching\thanks{I thank the Editor Rani Spiegler and three anonmyous referees for feedback that has greatly improved this paper. I am indebted to Eddie Dekel, Jeff Ely, and Alessandro Pavan for many fruitful conversations and for their continuing guidance and support. I also wish to thank H\'ector Chade, Federico Echenique, Jan Eeckhout, Guillaume Haeringer, John Hatfield, George Mailath, Pablo Schenone, James Schummer, Vasiliki Skreta, Alex Teytelboym, Asher Wolinsky, Leeat Yariv, Vijay Vazirani, Charles Zheng, and especially, Erik Eyster  and Jacob Leshno for useful discussions. All errors are, of course, my own.}}
\author{Laura Doval\thanks{Columbia University, New York, NY 10027. E-mail: \href{mailto:laura.doval@columbia.edu}{laura.doval@columbia.edu}}}

\begin{document}
\maketitle
\begin{abstract}
I introduce a stability notion, \emph{dynamic stability}, for two-sided dynamic matching markets where (i) matching opportunities arrive over time, (ii) matching is one-to-one, and (iii) matching is irreversible. The definition addresses two conceptual issues. First, since not all agents are available to match at the same time, one must establish which agents are allowed to form blocking pairs. Second, dynamic matching markets exhibit a form of externality that is not present in static markets: an agent's payoff from remaining unmatched cannot be defined independently of what other contemporaneous agents' outcomes are.
Dynamically stable matchings always exist. Dynamic stability is a necessary condition to ensure timely participation in the economy by ensuring that agents do not strategically delay the time at which they are available to match.
\end{abstract}
\section{Introduction}\label{sec:intro}
I formulate a notion of stability, denoted \emph{dynamic stability}, for two-sided dynamic matching markets where (i) matching opportunities arrive over time, (ii) matching is one-to-one, and (iii) matching is irreversible. Stability notions provide an analyst with a set of predictions for the self-enforcing outcomes of decentralized matching markets that depend only on the primitive payoff structure. While stability notions are extensively used in the study of static matching markets, they have not been systematically studied for dynamic matching markets, even though the latter are ubiquitous and cover many important applications, such as labor markets and child adoption. 
%
%
%
%

Defining stability in a dynamic matching market brings forth two new challenges that arise when taking into account agents' intertemporal incentives. First, since not all agents are available to match at the same time, it is natural to ask which pairs of agents can object to a proposed matching. Dynamic stability assumes that only agents who are available to match at the same time can form a blocking pair. Second, whether an agent finds their matching partner acceptable depends on what their value of remaining unmatched is. In turn, this value depends on what matching the agent \emph{conjectures} would ensue upon their decision to remain unmatched. Given a conjectured continuation matching, one could define an agent's acceptable partners to be those who are preferred to the continuation matching. This, together with a specification of the set of blocking pairs, is enough to determine whether a matching is stable in the dynamic economy: it should have no blocking pairs and agents should always be matched to acceptable partners.

The missing step is then to determine what matching the agent conjectures would result following their decision to remain unmatched. The first difficulty is that the set of agents available to match from tomorrow onward depends both on the arrivals into the economy and on who remains unmatched from previous periods. In other words, today's matching together with tomorrow's arrivals define the set of feasible continuation matchings. When contemplating remaining unmatched, the agent then needs to conjecture both who else remains unmatched today and tomorrow's continuation matching. Thus, as in the literature on the core with externalities (see, for instance, \cite{shapley1969core,rosenthal1971external,richter1974core,sasaki1996two,pycia2017matching,rostek2017matching}), an agent's payoff from remaining unmatched cannot be defined independently of what other contemporaneous agents' matching outcomes are. This externality sets apart dynamic matching markets from their static counterparts.

 Given a conjecture about who else remains unmatched today, not all continuation matchings are equally \emph{reasonable}. Indeed, the agent should correctly anticipate that the continuation matching should be itself self-enforcing. Thus, for a given conjecture about today's matching outcome, the agent rules out those continuation matchings that are not self-enforcing. This is still not enough to pin down a unique continuation matching and, thus, the value of remaining unmatched. For a given conjecture about who else remains unmatched today, there can be many self-enforcing matchings. Moreover, there can be many conjectures about who else remains unmatched today.
Thus, the last step in determining whether the agent finds their matching partner acceptable
%
is an assumption on how the agent selects amongst the reasonable conjectures. Like \cite{sasaki1996two}, I assume the agent prefers their matching partner to remaining unmatched if the agent prefers their matching partner to one of the conjectured continuation matchings.
Unlike \cite{sasaki1996two}, the agent does not entertain all continuation matchings but only those that are self-enforcing in the continuation economy.

 Dynamic stability (\autoref{definition:ds}) is a recursive definition that builds on the elements described above. A matching for the dynamic economy is dynamically stable if (i) there is no pair of agents who are available to match at the same time who prefer to match together and (ii) there is no agent who is matched to someone who is unacceptable. Contrary to static notions of stability, the set of acceptable partners today is defined using the set of dynamically stable matchings from tomorrow onward. Dynamically stable matchings always exist in any finite horizon economy (\autoref{theorem:stoch-existence}); I discuss their properties in \autoref{sec:main}. As I explain in \autoref{sec:main}, the proof of \autoref{theorem:stoch-existence} 
builds on the insight in \cite{sasaki1996two} that an agent's most pessimistic conjecture can be used to define an artificial economy without externalities where (static) stable matchings are known to exist. 

Dynamic matching markets pose new challenges for market design, two of which are analyzed in \autoref{sec:timing}. First, in a dynamic matching market, agents do not only choose whether to participate, but also \emph{when} to participate. \autoref{prop:stoch-queueing} shows that dynamic stability is a necessary condition for \emph{timely} participation in the market: whenever a matching fails to be dynamically stable, market participants have an incentive to delay the time at which they are available to match.\footnote{The mechanism design literature on revenue management has studied the problem of strategic participation (see, for instance, \cite{gershkov2015efficient,garrett2016intertemporal,bergemann2019progressive}).}  This echoes the observation in static matching markets that stability is a necessary condition for participation (\cite{roth1984evolution}). Second, in applications like school choice, college admissions, and teacher assignment, assignments are performed sequentially because not all agents are available to match at the same time (e.g., \cite{westkamp2013analysis,andersson2018sequential,dur2019sequential}). The centralized markets that perform these assignments operate via \emph{spot} mechanisms, which take the set of currently available agents and their rankings and outputs a matching as a function of the reported preferences, but do not condition on future matching possibilities. Unsurprisingly, spot mechanisms do not necessarily induce dynamically stable matchings (\autoref{ex:college-admissions-bis}). This raises the question of what outcomes may arise from employing spot mechanisms to perform assignments in a dynamic economy. \autoref{theorem:ds-implementation} shows that only dynamically stable matchings can arise as outcomes of subgame perfect Nash equilibrium of the non-cooperative game induced by a sequence of spot mechanisms. Thus, agents' forward looking behavior is enough to overcome the mechanism's inability to condition on all parameters of the economy. Achieving dynamically stable matchings, however, comes at the cost of truthful behavior: agents have an incentive to truncate their preferences above and beyond what they would do in a static economy (\cite{roth1991incentives}) so that their assignment reflects what they could have instead obtained by waiting to be matched.
\paragraph{Related literature} Outside of the literature on the core with externalities, the paper relates to five other strands of literature. The first strand is the literature on market design, which studies dynamic matching markets such as those in this paper, but from the point of view of optimality instead of stability (\cite{unver2010dynamic,anderson2015dynamic,leshno2017dynamic,schummer2015influencing,bloch2017dynamic,ashlagi2018matching,thakral2019matching,akbarpour2020thickness,arnosti2020design,baccara2020optimal}).\footnote{An exception is \cite{altinok2019dynamic} who studies stability in dynamic many-to-one matching markets.}
The present study of stability is important because stability is considered a key property for the success of algorithms (\citealp{roth1991natural}) and because it highlights the potential issues in applying the static notions of stability to dynamic environments.

The second strand is the literature on matching with frictions, which studies dynamic matching markets such as those in this paper in a non-cooperative framework (see \cite{burdett1997marriage,eeckhout1999bilateral,adachi2003search,lauermann2014stable} for non-transferable utility, and \cite{shimer2000assortative} for transferable utility). 
Like in this strand of the literature, an agent's value of remaining unmatched (their continuation value) is determined endogenously by the remaining agents in the market and the future matching opportunities.
 
%
The third strand studies stability notions for markets where matching opportunities are fixed and pairings can be revised over time (e.g., \cite{damiano2005stability,kurino2009credibility,kadam2018multiperiod,liu2018stability,kotowski2019perfectly}). 
The contribution relative to this strand is to provide stability notions for markets where matching opportunities arrive over time and matching is irreversible. As discussed in the introduction, that matching opportunities arrive over time and matching is irreversible introduces a primitive externality that is absent from these papers and must be addressed when defining what stability means. In particular, while this paper shares with \cite{liu2018stability} and \cite{kotowski2019perfectly} the perfection requirement (see also, \citealp{doval2015theory}) and with \cite{kotowski2019perfectly} the use of the approach pioneered by \cite{sasaki1996two}, the motivation for using this approach is different. While the externality in my paper is an intrinsic feature of the environment that must be addressed by the stability notion, the externality in \cite{kotowski2019perfectly} is a consequence of the perfection requirement in his stability notion, together with the assumption of non-separable payoffs.

The fourth strand studies sequential assignment problems (e.g., \cite{westkamp2013analysis,dogan2018does,andersson2018sequential,dur2019sequential,haeringer2019gradual,mai2019stability,feigenbaum2020dynamic}). Of these, only \cite{westkamp2013analysis,andersson2018sequential,dur2019sequential,mai2019stability} study models where not all agents are available to be matched at the same time. Their focus is, however, on the properties of the matching implemented by the mechanism from the point of view of stability in a static market. \autoref{theorem:ds-implementation} echoes observations in \cite{westkamp2013analysis}, \cite{andersson2018sequential}, and \cite{dur2019sequential} that sequential assignment may be at odds with stability or truthful behavior. \autoref{theorem:ds-implementation} complements these results by identifying dynamic stability as the solution concept for sequential assignment problems. Since stability notions are oftentimes used for preference identification,
\autoref{theorem:ds-implementation} can inform the empirical study of sequential assignment problems as in \cite{narita2018match} and \cite{neilson2020aftermarket}.
%

The fifth strand is the literature on the farsighted stable set, which is used to model externalities in coalition formation games (e.g., \cite{harsanyi1974equilibrium,chwe1994farsighted,mauleon2011neumann,ray2015farsighted}). As in this literature, agents in my model understand the terminal consequences of their moves. While farsighted stability focuses on the credibility of coalitional blocks, dynamic stability focuses on the credibility of the continuation matchings used to dissuade agents from blocking.\footnote{Relatedly, \cite{ray1997equilibrium} provide a recursive definition of binding agreements in a static model where a blocking coalition anticipates, amongst other things, that the agents outside the coalition form a binding agreement amongst themselves. \cite{acemoglu2012dynamics} study a dynamic noncooperative game where in each period a winning coalition can change the prevailing state. They show that in the patient limit equilibrium outcomes satisfy a (static) solution concept in the spirit of farsighted stability, which they also denote dynamic stability.} 

 \paragraph{Organization} The rest of the paper is organized as follows. 
 \autoref{sec:model} describes the model and \autoref{sec:dynamic-stability} defines dynamic stability. \autoref{sec:main} shows that dynamically stable matchings exist and discusses their properties. \autoref{sec:timing} studies participation and incentives in dynamic matching markets. 
 All proofs are in the appendices.
\section{Model}\label{sec:model}
The economy lasts for $\terminal<\infty$ periods. There are two sides $A$ and $B$. Agents on side $A$ are labeled $\agenta\in\sidea$, while agents on side $B$ are labeled $\agentb\in\sideb$, where $\sidea,\sideb$ are finite sets.

 For any $t\geq 1$, let $\economy^t$ denote the subset of $(2^\sidea\times2^\sideb)^t$ that satisfies the following property. A tuple $E^t=(A_1,B_1,\dots,A_t,B_t)\in\economyt$ only if $A_i\cap A_j=\emptyset$ and $B_i\cap B_j=\emptyset$, whenever $i\neq j$.\footnote{To economize on notation, I assume that no two agents with the same characteristic arrive within or across periods. However, none of the results depend on this. See \autoref{sec:general}.} An economy of length $T$ is a distribution $\Gterminal$ on $\mathcal{E}^\terminal$. In what follows, I refer to an element $\economyrt$ of \economyt\ as a realization, to a tuple $(A_1,\dots,A_t)$ as a side-A arrival, and to a tuple $(B_1,\dots,B_t)$ as a side-B arrival.

Given a realization $\economyrt=(A_1,B_1,\dots,A_t,B_t)$, let $\arrivalaupto_\dateindex(\economyrt)=\cup_{t^\prime=1}^\dateindex A_{t^\prime}$ denote the implied arrivals on side $A$ through period $\dateindex\leq t$; similarly, let $\arrivalbupto_\dateindex(\economyrt)=\cup_{t^\prime=1}^\dateindex B_{t^\prime}$ denote the implied arrivals on side $B$ through period $\dateindex\leq t$. Finally, if $\dateindex\leq t$ and $\economyrt=(E^s,E^{t-s})$, then I say that $E^\dateindex$ is a \emph{truncation} of $\economyrt$ and that $\economyrt$ \emph{follows} $\economyr^\dateindex$.

 Fix an economy \Gterminal. Definitions \ref{definition:matching-t} and \ref{definition:matching} below define the set of feasible allocations for \Gterminal.\
\begin{definition}\label{definition:matching-t}
A period-$t$ matching for realization \economyrt\ is a mapping
\begin{align*}
\matchingt:\arrivalauptot(\economyrt)\cup\arrivalbuptot(\economyrt)\mapsto\arrivalauptot(\economyrt)\cup\arrivalbuptot(\economyrt)
\end{align*}
such that
\begin{enumerate}
\item For all $\agenta\in\arrivalauptot(\economyrt)$, $\matchingt(\agenta)\in\{\agenta\}\cup\arrivalbuptot(\economyrt)$,
\item For all $\agentb\in\arrivalbuptot(\economyrt)$, $\matchingt(\agentb)\in\arrivalauptot(\economyrt)\cup\{\agentb\}$,
\item For all $k\in\arrivalauptot(\economyrt)\cup\arrivalbuptot(\economyrt)$, $\matchingt(\matchingt(k))=k$.
\end{enumerate}
\end{definition}
\begin{definition}\label{definition:matching}
A matching \matching\ for \Gterminal\ is a map on $\cup_{t=1}^\terminal\economyt$ such that
\begin{enumerate}
\item For all $t\in\{1,\dots,\terminal\}$, for all $\economyrt\in\economyt$, $\matching(\economyrt)$ is a period-t matching,
\item For all $t\in\{1,\dots,T\}$, for all $\economyrt\in\economyt$, for all $\agenta\in\arrivalauptot(\economyrt)$, if $\matching(\economyrt)(a)\neq a$, then $\matching(E^s)(a)=\matching(\economyrt)(a)$ for all $\dateindex\geq t$ and $E^s$ that follow $\economyrt$.
\end{enumerate}
\end{definition}
Part 2 of \autoref{definition:matching} incorporates the idea that matchings in the economy are irreversible. Let \Matchingsfinal\ denote the set of matchings for an economy of length \terminal.\footnote{The particular economy of length $\terminal$, i.e., the distribution $\Gterminal$ on $\economyterminal$, defines which arrival sequences on the domain of \matching\ have positive probability.} Finally, given a matching $\matching\in\Matchings_\terminal$ and a realization \economyrt,\ let $\Matchings_\terminal(\matching,\economyrt)$ denote the subset of $\Matchings_\terminal$ that coincide with \matching\ at realizations $\economyrs$ that do not follow \economyrt.\

Fix a matching $\matching$ and a realization of the arrivals through period $t$, $\economyrt=(E^{t-1},(A_t,B_t))$. This determines the agents who can match at \economyrt:\
\begin{align*}
\cala(\matchinguptot,\economyrt)&=\{a\in\arrivalaupto_{t-1}(\economyrt):\matching(E^{t-1})(a)=a\}\cup A_t\\
\calb(\matchinguptot,\economyrt)&=\{b\in\arrivalaupto_{t-1}(\economyrt):\matching(E^{t-1})(b)=b\}\cup B_t.
\end{align*}
In the above definition, $\matchinguptot$ is defined as $(\matching(E^1),\dots,\matching(E^{t-1}))$ where for each $\dateindex\leq t-1$, $E^s$ is a truncation of $\economyrt$.

 A matching \matching\ and a realization $\economyrt$ also define a continuation economy of length $\terminal-t$, where the distribution of arrivals $G_{\terminal-t}(\matching^t,\economyrt)$ assigns probability \[\Gterminal((A_{t+1},B_{t+1},\dots,A_\terminal,B_\terminal)|\economyrt)\] to the length $\terminal-t$ arrival \[(\cala(\matching^t,(\economyrt,A_{t+1},B_{t+1})),\calb(\matching^t,(\economyrt,A_{t+1},B_{t+1})),\dots,A_\terminal,B_\terminal),\] and $0$ to all others.

 I close the model by defining agents' preferences. Each $\agenta\in\sidea$ defines a discount factor $\delta_\agenta\in[0,1)$ and a Bernoulli utility, $u(\agenta,\cdot):\sideb\cup\{\agenta\}\mapsto\mathbb{R}$. Similarly, each $\agentb\in\sideb$ defines a discount factor $\delta_\agentb\in[0,1)$ and a Bernoulli utility, $v(\cdot,\agentb):\sidea\cup\{b\}\mapsto\mathbb{R}$. I assume that for all $\agenta\in\sidea$ and all $\agentb\in\sideb$, $u(\agenta,\agenta)=v(\agentb,\agentb)=0$.

 Fix a matching \matching\ and a realization $\economyrt$. Let $a\in\cala(\matchinguptot,\economyrt)$. For each $E^\terminal=(\economyrt,\cdot)$, let $t_\matching(a,E^\terminal)$ denote the first date at which \agenta\ is matched under \matching. That is, the smallest $t\leq\dateindex$ such that $m(E^s)(a)\neq a$ and $E^s$ is a truncation of $E^\terminal$; otherwise, let $t_\matching(a,E^\terminal)=\terminal$. Then, let
\begin{align*}
U(a,\matching,\economyrt)=\mathbb{E}_{G(\cdot|\economyrt)}[\delta_a^{t_\matching(a,\economyrt,\cdot)-t}u(a,\matching(\economyrt,\cdot)(a))],
\end{align*}
denote $a$'s payoff from matching \matching\ at date $t$ when the realization is $\economyrt$. Similarly, for $b\in\calb(\matchinguptot,\economyrt)$, let
\begin{align*}
V(b,\matching,\economyrt)=\mathbb{E}_{G(\cdot|\economyrt)}[\delta_b^{t_{\matching}(b,\economyrt,\cdot)-t}v(\matching(\economyrt,\cdot)(b),b)],
\end{align*}
denote $b$'s payoff from matching \matching\ at date $t$ when the realization is $\economyrt$.

I record for future reference two properties that matchings \matching\ may satisfy: individual rationality (\autoref{definition:ir}) and stability for static markets (\autoref{definition:static-stability}). 
\begin{definition}\label{definition:ir}
The matching \matching\ for economy \economyrterminal\ is \emph{individually rational} if if for all realizations $\economyrterminal$, for all $a\in\arrivalauptoterminal(\economyrterminal)$, and all $b\in\arrivalbuptoterminal(\economyrterminal)$, $u(a,\matching(\economyrterminal)(a))\geq 0$ and $v(\matching(\economyrterminal)(b),\agentb)\geq 0$.
\end{definition}
That is, matching \matching\ is individually rational if each agent prefers their matching partner to remaining unmatched through period \terminal.
\begin{definition}\label{definition:static-stability}[\citealp{gale1962college}]
The matching $\matching_1$ for realization $\economyr^1=(\arrivala_1,\arrivalb_1)$ is \emph{stable} if the following hold:
\begin{enumerate}[leftmargin=*,label=(\subscript{S}{{\arabic*}})]
\item\label{itm:ss2} For all $\agenta\in\arrivala_1$, $U_1(\agenta,\matching_1)\geq 0$,
\item\label{itm:ss3} For all $\agentb\in\arrivalb_1$, $V_1(\agentb,\matching_1)\geq 0$,
\item\label{itm:ss1} There is no pair $\pair\in\arrivala_1\times\arrivalb_1$ such that $\genericu>U_1(\agenta,\matching_1)$ and $\genericv>V_1(\agentb,\matching)$.
\end{enumerate}
Let $\stable(\economyr^1)$ denote the set of stable matchings for $\economyr^1$.
\end{definition}
In a one-period economy, matching $\matching_1$ is stable if each agent prefers their matching partner to remaining unmatched and
 there is no pair of agents who prefer each other to the partners assigned by $\matching_1$.

Throughout, I use the following example to illustrate the concepts in the paper:
\begin{example}\label{example:main-example}
The economy lasts for two periods, i.e., \terminal=2, and \Gtwo\ assigns probability $1$ to one realization, \economyrtwo. In \periodone, \erdos, \kuhn, and \gale\ arrive on side $A$, while \renyi\ and \tucker\ arrive on side $B$. In \periodtwo, \shapley\ and \nash\ arrive on side $B$, and there are no arrivals on side $A$. That is, 
\begin{align*}
\arrivala_1=\{\erdos,\kuhn,\gale\},\arrivalb_1=\{\renyi,\tucker\},\arrivala_2=\{\emptyset\},\arrivalb_2=\{\shapley,\nash\}.
\end{align*}
Below, I list the agents' preferences. If $(\shapley, 1)$ (resp., $(\nash,1)$) appears before $(\tucker, 0)$ in an agent's ranking, then they prefer to wait $1$ period to match with \shapley\ (resp., \nash) over matching immediately with \tucker. That is, the $0$s and $1$s are the exponent of the discount factors, and the list provides the ranking of the discounted utilities, $\{\delta_\agenta^{t-1}\genericu\}$. For side $B$, the list represents the ranking of utilities, $\{\genericv\}$. 
\small
\[\hspace{-1cm}\left.\begin{array}{lllll}
\erdos:&(\shapley,1)&(\renyi,0)&\hspace{-0.25cm}(\nash\ ,0)&\\
\kuhn:&(\shapley,1)&(\nash\ ,0)&(\tucker,0)&\hspace{-0.3cm}(\nash,1)\\
\gale:&(\tucker,0)&(\shapley,1)&&\\
&&&&\end{array}
\right.\hspace{0.25cm}\left.\begin{array}{llll}
\renyi:&\erdos&&\\
\tucker:&\kuhn&\gale&\\
\shapley:&\gale&\erdos&\kuhn\
\\
\nash:&\kuhn&\erdos&\end{array}
\right.\]
\normalsize
Consider the following three matchings (In what follows, a horizontal line separates matchings that occur in different periods):

\begin{figure}[h!]
{\small{\[\hspace{-1cm}\matching^L=\left(\begin{array}{rcl}\erdos&\_\_&\renyi\\\gale&\_\_&\tucker\\\hline\kuhn&\_\_&\shapley\\\emptyset&\_\_&\nash\end{array}\right) \hspace{0.25cm}
\matching^C=\left(\begin{array}{rcl}\erdos &\_\_&\renyi\\\hline\kuhn&\_\_&\tucker\\\gale&\_\_&\shapley\\\emptyset&\_\_&\nash\end{array}\right)\hspace{0.25cm}
\matching^R=\left(\begin{array}{rcl}\gale&\_\_&\tucker\\\hline\erdos&\_\_&\shapley\\\kuhn&\_\_&\nash\\\emptyset&\_\_&\renyi\end{array}\right)\]
}}
\caption{Three matchings for the economy in \autoref{example:main-example}.}\label{fig:ds-pairwise}  
\end{figure}
\autoref{fig:ds-pairwise} illustrates for each matching only the pairings that happen within each period. For instance, note that $\matching^L(\economyr^1)$ specifies both that \erdos\ matches with \renyi\ and \gale\ matches with \tucker\ and that \kuhn\ is single. That is, $\matching^L(\economyr^1)(\kuhn)=\kuhn$. Similarly, $\matching^L(\economyrtwo)$ specifies that \kuhn\ matches with \shapley\ and also records the \periodone-matching, $\matching^L(\economyrtwo)(\erdos)=\renyi,\matching^L(\economyrtwo)(\gale)=\tucker$.

Each of the matchings in \autoref{fig:ds-pairwise} satisfies \autoref{definition:ir}: each agent's matching partner is preferred to remaining single in both periods. However, only the period-2 matchings $\matching^L(\economyrtwo)$ and $\matching^R(\economyrtwo)$ satisfy \autoref{definition:static-stability}. To see this, consider $\matching^C$. If agents match according to $\matching^C(\economyrone)$ in \periodone, this induces a one-period economy in \periodtwo, with agents on side \arrivala, $\cala(\matching_1^C,\economyrtwo)=\{\kuhn,\gale\}$, and agents on side \arrivalb, $\calb(\matching_1^C,\economyrtwo)=\{\tucker,\shapley,\nash\}$. Note that $\matching^C(\economyrtwo)$ does not satisfy \autoref{definition:static-stability} for this one-period economy: \kuhn\ and \nash\ form a block. Instead, it is immediate to verify that $\matching^L(\economyrtwo)$ (resp., $\matching^R(\economyrtwo)$) does satisfy \autoref{definition:static-stability} for the one-period economy, $\Gone(\matching^L,\economyrone)$ (resp., $\Gone(\matching^R,\economyrone)$).
\end{example}
\section{Dynamic Stability}\label{sec:dynamic-stability}
\autoref{sec:dynamic-stability} defines dynamic stability (\autoref{definition:ds}). Like the static notion of stability, dynamic stability is defined by the absence of pairwise blocks and the requirement that each agent is matched to an acceptable partner, i.e., someone that is preferred to remaining unmatched. Unlike the static notion of stability, the value of remaining unmatched is determined endogenously (see \autoref{eq:conjectures} below). In what follows, I first introduce the definition, using \autoref{example:main-example} to illustrate its components. I then discuss in detail the different elements in the definition.

Given a matching \matching, the goal is to determine whether the agents will follow the prescription of \matching\ in every period $t\in\{1,\dots,\terminal\}$ and realization \economyrt. Fix a period $t$ and a realization \economyrt. Suppose that the agents have matched according to \matching\ through period $t-1$. Thus, the set of agents who can match in period $t$ is $\cala(\matchinguptot,\economyrt)\cup\calb(\matchinguptot,\economyrt)$. 

There are two reasons why the agents in  $\cala(\matchinguptot,\economyrt)\cup\calb(\matchinguptot,\economyrt)$ may prefer not to follow the prescription of \matching. First, there could be a pair of agents who prefer to match together over matching according to \matching. Second, there could be a single agent who would prefer not to match according to \matching. The latter could be either because \matching\ prescribes that the agent matches at \economyrt, but the agent could do better by remaining unmatched, or because \matching\ prescribes that the agent is unmatched at \economyrt, but the continuation matching, $(\matching(\economyr^{t+\dateindex}))_{\dateindex=1}^{\terminal-t}$ is not self-enforcing. In either case, I say that the agent prefers to be unavailable to match at \economyrt. Note that an agent who is unavailable to match in period $t$ automatically remains unmatched in period $t$, but an agent who is available to match in period $t$ could still remain unmatched at the end of period $t$.\footnote{As I explain after the statement of \autoref{definition:ds}, it is enough to consider only objections by agents who in period $t$ are supposed to match under \matching. However, considering both kinds of objections allows me to introduce a key property of dynamically stable matchings that plays a role in the proof of \autoref{theorem:stoch-existence} (see \autoref{rem:consistency}).}

An agent can determine whether they prefer to match according to \matching\ at \economyrt\ or with a contemporaneous partner on the other side by comparing their payoff under \matching\ with their payoff from matching with the new partner. Instead, the matching \matching, together with the model primitives, do not provide enough information to determine whether an agent should be available to match at \economyrt. 

%
To see this, recall the three matchings in \autoref{fig:ds-pairwise}, reproduced below:

\begin{figure}[h!]
{\small{\[\hspace{-1cm}\matching^L=\left(\begin{array}{rcl}\erdos&\_\_&\renyi\\\gale&\_\_&\tucker\\\hline\kuhn&\_\_&\shapley\\\emptyset&\_\_&\nash\end{array}\right) \hspace{0.25cm}
\matching^C=\left(\begin{array}{rcl}\erdos &\_\_&\renyi\\\hline\kuhn&\_\_&\tucker\\\gale&\_\_&\shapley\\\emptyset&\_\_&\nash\end{array}\right)\hspace{0.25cm}
\matching^R=\left(\begin{array}{rcl}\gale&\_\_&\tucker\\\hline\erdos&\_\_&\shapley\\\kuhn&\_\_&\nash\\\emptyset&\_\_&\renyi\end{array}\right)\]
}}
\caption{Matchings in \autoref{fig:ds-pairwise}}\label{fig:ds-pairwise-bis}  
\end{figure}

Note that $\matching^C$ and $\matching^R$ can be ruled out as predictions for the economy in \autoref{example:main-example} using only the information contained in the matching. First, as argued in \autoref{sec:model}, $\matching^C(\economyrtwo)$ is not stable in \periodtwo. Thus, even if agents match according to $\matching^C(\economyrone)$ in \periodone, $\matching^C(\economyrtwo)$ cannot be enforced in \periodtwo. Second, $\matching^R$ can also be ruled out as a prediction for the economy: \kuhn\ and \tucker\ prefer each other to their outcome under $\matching^R$. Thus, they would match together in \periodone\ if they anticipate that $\matching^R$ is the suggested outcome for the economy.

Instead, $\matching^L$ cannot be ruled out using the information contained in the matching alone. In this matching, there is only one pair who prefers each other to their outcome under $\matching^L$:
\erdos\ and \shapley. However, under $\matching^L$, \erdos\ and \shapley\ never meet: \erdos\ is supposed to match with \renyi\ and exit before \shapley\ arrives. Thus, for \erdos\ to match with \shapley, \erdos\ must prefer to wait for \shapley\ to arrive over matching with \renyi. Deciding whether waiting for \shapley\ is preferred to matching with \renyi\ in \periodone\ depends on what \erdos\ expects the matching would be should he remain unmatched in \periodone. However, $\matching^L$ says nothing about what would happen if \erdos\ chooses to be unavailable to match in \periodone.
%
Hence, he cannot decide whether he prefers matching with \renyi\ to remaining unmatched in \periodone.

In the dynamic economy, together with the matching \matching\ one should at least prescribe for each period $t$, realization \economyrt, and for each agent that can match at \economyrt, the matching that would ensue if the agent chose to be unavailable to match. This is the role of what I dub the agent's \emph{conjectures}, to which I turn next.

Consider an agent \agent\ who can match at \economyrt\ under \matching, and instead considers being unavailable to match at \economyrt. A conjecture for agent \agent\ is a matching \matchingb\ such that $\matchingb(\economyrt)(\agent)=\agent$. Since \agent\ cannot alter the matchings through period $t-1$, $\matchingb$ must be an element of $\Matchingsfinal(\matching,\economyrt)$. Note that the conjecture defines both who matches at \economyrt, $\matchingb(\economyrt)$, and the continuation matching $(\matchingb(\economyr^{t+\dateindex}))_{\dateindex=1}^{\terminal-t}$. The reason is that to determine the set of feasible matchings from period $t+1$ onward should \agent\ decide to remain unmatched at \economyrt, one needs to determine who else remains unmatched in period $t$. After all, the unmatched agents at \economyrt\ \emph{together} with the new arrivals defines the set of agents available to match from period $t+1$ onward.

Having determined the set of \emph{feasible} conjectures for \agent, I identify in two steps which matchings \agent\ can exclude from consideration. First, note that each conjecture \matchingb,  determines a continuation economy $G_{\terminal-t}(\matchingb^t,\economyrt)$. Because the remaining unmatched agents and the new entrants can (and will) object to any matching that is not self-enforcing, \agent\ should exclude from consideration any such matching. Recall that dynamic stability is a recursive definition and let $\ds_t$ denote the correspondence that maps economies of length $t$, \Gt, to the set of dynamically stable matchings for \Gt.
%
Then, \agent\ can exclude from consideration any matching \matchingb\ such that  $(\matchingb(\economyr^{t+\dateindex}))_{\dateindex=1}^{\terminal-t}$, that is not an element of $\ds_{\terminal-t}\left(G_{\terminal-t}(\matchingb^t,\economyrt)\right)$.\footnote{There is a slight abuse of notation. While $\ds_{\terminal-t}\left(G_{\terminal-t}(\matchingb^t,\economyrt)\right)$ defines a matching only for the agents that are unmatched from $t+1$ onward, \matchingb\ also specifies the outcome for those who have matched through period $t$. Thus, this should be read as ``$(\matchingb(\economyr^{t+\dateindex}))_{\dateindex=1}^{\terminal-t}$ coincides with an element of $\ds_{\terminal-t}$ for the agents who are \emph{yet} to be matched at the end of period $t$.''} Second, one needs to determine what period-$t$ matchings (if any) \agent\ can exclude from consideration. The minimal restriction on the set of period-$t$ matchings that \agent\ should entertain is that the matching formed by the agents who match in period $t$ forms a (static) stable matching. Thus, \agent\ never conjectures that the agents who exit the economy in period $t$ could have found a better matching amongst themselves. Formally, 
 \begin{definition}\label{definition:ds-static-stability} Fix a matching, \matching, a period $t$, and a realization \economyrt. $\matching(\economyrt)$, is \emph{stable amongst those who match in period $t$} if $\matching(\economyrt)\in\stable(A_t^\matching,B_t^\matching)$ where
 $A_t^\matching=\{\agenta\in\cala(\matchinguptot,\economyrt):\matching(\economyrt)(\agenta)\neq\agenta\}$, and similarly, $B_t^\matching=\{\agentb\in\calb(\matchinguptot,\economyrt):\matching(\economyrt)(\agentb)\neq\agentb\}$.
\end{definition}
With these restrictions at hand, one can define the set of continuation matchings that \agent\ conjectures may ensue if \agent\ decides to make themselves unavailable to match in period $t$. I denote this set by $M_\ds(k,\matching,\economyrt)$. Formally,
\begin{align}\label{eq:conjectures}
M_\ds(k,\matching,\economyrt)=\left\{\begin{array}{ll}\matchingb\in\Matchingsfinal(\matching,\economyrt):&\text{(i) }\matchingb(\economyrt)(\agent)=\agent,
\\
&\text{(ii) }(\matchingb(\economyr^{t+\dateindex}))_{\dateindex=1}^{\terminal-t}\in\ds_{\terminal-t}\left(G_{\terminal-t}(\matchingb^t,\economyrt)\right),\\
&\text{(iii) }\matchingb(\economyrt)\text{ satisfies \autoref{definition:ds-static-stability}}
\end{array}\right\}.
\end{align}
\begin{ex}[continued] I now illustrate the set $M_\ds(\agent,\matching,\economyrt)$ for the case in which $\agent=\erdos$ and $\matching=\matching^L$. 
Consider the following four matchings:

\begin{figure}[h!]
{\small{
\[\matchingb_{E1}=\left(\begin{array}{ccc}\gale&\_\_&\tucker\\\hline\kuhn&\_\_&\nash\\\erdos&\_\_&\shapley\\\emptyset&\_\_&\renyi\end{array}\right)\hspace{0.25cm}
\matchingb_{E2}=\left(\begin{array}{ccc}\kuhn&\_\_&\tucker\\\hline\gale&\_\_&\shapley\\\erdos&\_\_&\renyi\\\emptyset&\_\_&\nash\end{array}\right)\]
\[\matchingb_{E3}=\left(\begin{array}{ccc}\gale&\_\_&\tucker\\\hline\kuhn&\_\_&\shapley\\\erdos&\_\_&\renyi\\\emptyset&\_\_&\nash\end{array}\right)\hspace{0.25cm}
\matchingb_{E4}=\left(\begin{array}{ccc}\gale&\_\_&\tucker\\\kuhn&\_\_&\renyi\\\hline\erdos&\_\_&\shapley\\\emptyset&\_\_&\nash\end{array}\right)\]}}
\caption{Four conjectures for \erdos, only $\matchingb_{E1}$ and $\matchingb_{E2}$ are valid.}\label{fig:erdos-conjectures}
\end{figure}

In \autoref{fig:erdos-conjectures}, only $\matchingb_{E1},\matchingb_{E2}$ are valid conjectures for \erdos. As it will be clear after the statement of \autoref{definition:ds}, dynamic stability reduces to the static notion of stability in one-period economies. Since $\matchingb_{E3}$ does not prescribe a stable matching in \periodtwo\ (\erdos\ and \shapley\ are a block), then it fails to satisfy condition (ii) in \autoref{eq:conjectures}. Instead, $\matchingb_{E4}$ is not a valid conjecture for \erdos\ because the period-1 matching does not satisfy condition (iii) in \autoref{eq:conjectures}. Under $\matchingb_{E4}$, \kuhn\ and \tucker\ match in \periodone\ and would have preferred to match with each other over matching with their matching partners in $\matchingb_{E4}$, so that the period-$1$ matching does not satisfy \autoref{definition:ds-static-stability}. 
\end{ex}

\autoref{definition:ds} defines dynamic stability. Together with prescribing how to use the conjectures to define an agent's value of remaining unmatched, it also prescribes which pairs of agents can object to a given matching.
\begin{definition}\label{definition:ds}
Given the correspondences $(\ds_t)_{t\leq\terminal-1}$, matching \matching\ is \emph{dynamically stable} for \Gterminal\ if for all $t\in\{1,\dots,\terminal\}$ and all realizations \economyrt, the following hold:
\begin{enumerate}[leftmargin=*,label=(\subscript{DS}{{\arabic*}})]

\item\label{itm:ds2} For all $\agenta\in\cala(\matchinguptot,\economyrt)$, there exists $\matchingb\in M_\ds(\agenta,\matching,\economyrt)$ such that $U(\agenta,\matching,\economyrt)\geq U(\agenta,\matchingb,\economyrt)$,
\item\label{itm:ds3} For all $\agentb\in\calb(\matchinguptot,\economyrt)$, there exists $\matchingb\in M_\ds(\agentb,\matching,\economyrt)$ such that $V(\agentb,\matching,\economyrt)\geq V(\agentb,\matchingb,\economyrt)$,
\item\label{itm:ds1} There is no pair $\pair\in\cala(\matchinguptot,\economyrt)\cup\calb(\matchinguptot,\economyrt)$ such that $\genericu>U(\agenta,\matching,\economyrt)$ and $\genericv>V(\agentb,\matching,\economyrt)$.
\end{enumerate}
Let $\ds_\terminal(\Gterminal)$ denote the set of dynamically stable matchings for \Gterminal.
\end{definition}
In words, a matching \matching\ is dynamically stable if the following hold. First, for each period $t$, realization \economyrt, and each agent \agent\ who can match at \economyrt, there exist matchings $\matchingb_\agent\in M_\ds(\agent,\matching,\economyrt)$, such that \agent's payoff under \matching\ is at least the payoff \agent\ would obtain by remaining unmatched in period $t$ when the matching is $\matchingb_\agent$. Second, there is no pair of agents who can match at the same time and prefer matching together to matching according to \matching. Since $\cala(\matchinguptot,\economyrt)\cup\calb(\matchinguptot,\economyrt)$ may contain agents who arrive before period $t$, condition \ref{itm:ds1} does not rule out \emph{all} blocks involving agents who arrive in different periods. Instead, it requires that at the moment of blocking both agents are present, since only then they can evaluate whether they want to carry out the block. Since condition \ref{itm:ds1} plays the same role as condition \ref{itm:ss1} in static environments, in what follows, a matching \matching\ that satisfies condition \ref{itm:ds1} is denoted \emph{pairwise stable}.

When $\terminal=1$, the definition of dynamic stability reduces to the requirement that $\matching(\economyrone)$ satisfies the static notion of stability for each realization \economyrone. Indeed, fix a matching \matching\ for the one-period economy. Fix a realization \economyrone\ and an agent \agenta\ in $\arrivala_1$. Note that all matchings \matchingb\ in the set $M_{\ds}(a,\cdot)$ satisfy that $\matchingb(\economyrone)(a)=a$ (part (ii) in the definition of $M_{\ds}$ is vacuous). Thus, condition \ref{itm:ds2} simply states that $a$ prefers \matching\ to remaining single. Condition \ref{itm:ds1} in \autoref{definition:ds} states that the matching $\matching(\economyrone)$ has no pairwise blocks. Thus, when $\terminal=1$, \autoref{definition:ds} reduces to \autoref{definition:static-stability} for each \economyrone.
newline\indent
While the comparison of Definitions \ref{definition:static-stability} and \ref{definition:ds} suggests that in the dynamic economy the value of waiting to be matched replaces the value of being single, this analogy is potentially incomplete. In a static economy, an agent who is single under a stable matching receives \emph{exactly} the value they can guarantee by being single. Instead, \autoref{definition:ds} suggests that for an agent who is unmatched in period $t$, there may be a gap between the utility of matching \matching\ and the utility of the conjectured matching \matchingb, which represents the value the agent can guarantee by choosing to be unavailable to match. This gap is only superficial since dynamically stable matchings satisfy the following property, which I record for future reference:
\begin{rem}
\label{rem:consistency}
Let \matching\ be a dynamically stable matching for \Gterminal. For all $t\geq1$, for all realizations \economyrt, and for all $\agent\in\cala(\matchinguptot,\economyrt)\cup\calb(\matchinguptot,\economyrt)$ such that $\matching(\economyrt)(\agent)=\agent$, then $\matching\in M_\ds(\agent,\matching,\economyrt)$.
\end{rem}
Thus, for those agents who are unmatched under \matching\ in a given period, the conjectured matching, \matchingb, can be picked to exactly coincide with the matching \matching. That is, if \matching\ is dynamically stable an agent who can match at \economyrt, but remains unmatched at \economyrt\ under \matching, is indifferent between being available to match at \economyrt\ and being unavailable to match at \economyrt. It follows that if a matching \matching\ satisfies condition \ref{itm:ds1} and for each period $t$, and realization \economyrt, conditions \ref{itm:ds2} and \ref{itm:ds3} for agents \agent\ such that $\matching(\economyrt)(\agent)\neq\agent$, then \matching\ is dynamically stable.

I now illustrate \autoref{definition:ds} for the case in which $\terminal=2$ using \autoref{example:main-example}:
\addtocounter{ex}{-1}
\begin{ex}[continued]
There are two dynamically stable matchings: 
\label{main-example-ds}
\begin{figure}[h!]
\[\matching^L=\left(\begin{array}{lcl}\erdos&\_\_&\renyi\\\
\gale&\_\_&\tucker\
\\
\hline
\kuhn&\_\_&\shapley\
\\
\emptyset&\_\_&\nash\ 
\end{array}\right)\hspace{1cm}\matchingb^{L}=\left(\begin{array}{lcl}\erdos&\_\_&\renyi\\\
\kuhn&\_\_&\tucker\
\\
\hline
\gale&\_\_&\shapley\
\\
\emptyset&\_\_&\nash\ 
\end{array}\right)
\]
\caption{Dynamically stable matchings in \autoref{example:main-example}.}\label{fig:ex-ds}
\end{figure}

To check that $\matching^L$ is dynamically stable, it only remains to verify that \erdos\ cannot improve on his matching outcome by waiting to be matched. Since $\matchingb_{E2}$ in \autoref{fig:erdos-conjectures} is a reasonable conjecture for \erdos, it follows that by choosing to remain unmatched in \periodone, he can guarantee, at most, his payoff from matching with \renyi\ in \periodtwo.\ Therefore, \erdos\ prefers to match with \renyi\ in \periodone\ rather than to wait for \shapley\ to arrive. Note how the requirement that \erdos's conjectured matching induces a stable matching in period $2$ ``prevents'' \erdos\ and \shapley\ from blocking $\matching^L$. Under the conjectured matching $\matchingb_{J2}$, \shapley\ matches with \gale, who is preferred to \erdos. Indeed, this is the only stable matching in \periodtwo. Thus, when \erdos\ and \shapley\ cannot agree in advance that they will match together in \periodtwo, there are instances in which \shapley\ is not willing to match with \erdos, once \erdos\ waits for \shapley\ to arrive. Anticipating this, \erdos\ prefers to match with \renyi\ in \periodone.

 The matching $\matchingb^L$ is also dynamically stable. Indeed, the matching $\matchingb_{E1}$ in \autoref{fig:erdos-conjectures} is a valid conjecture for \kuhn\ when he considers remaining single in \periodone. Since \kuhn\ prefers to match with \tucker\ in \periodone\ over matching with \nash\ in \periodtwo, he cannot object to $\matchingb^L$ by waiting to be matched.
\end{ex}
\textbf{Discussion:}\label{discussion} Having introduced the definition of dynamic stability, I now discuss the role of its different components, starting with the set of conjectures, $M_\ds(\agent,\matching,\economyrt)$.

Two assumptions about the timing of an agent's decision to match in period $t$ according to matching \matching\ underlie the definition of the set of conjectures, $M_\ds(\agent,\matching,\economyrt)$:

First, this decision is made before agents are matched according to the prescribed matching in period $t$. To see this, note that \agent\ does not necessarily assume that the remaining agents, except perhaps for \agent's matching partner, match according to $\matching(\economyrt)$. As an illustration, consider \autoref{example:main-example}. When considering whether \erdos\ could block $\matching^L$, both $\matchingb_{E1}$ and $\matchingb_{E2}$ in \autoref{fig:erdos-conjectures} are valid conjectures under dynamic stability. However, only $\matchingb_{E1}$ respects that \gale's assignment is the same as in $\matching^L$. It is well-known in static matching models with externalities that if agents make decisions to block only after agents are matched according to \matchingt, stable matchings may not exist (\citealp{chowdhury2004marriage}). \cite{doval2015theory} shows that this observation extends to dynamic matching markets. 

Second, all agents who can match in period $t$, i.e., the remaining unmatched agents together with the new arrivals, decide \emph{simultaneously} whether they will be available to match in period $t$ (and also whether to form blocking pairs). That is, when agent \agent\ contemplates being unavailable to match in period $t$, two things are true: (i) agent \agent\ has not yet made themselves unavailable to match, and (ii) no one else has yet made this decision. 

To understand the role of the second assumption in \autoref{eq:conjectures}, note the following. The requirement that a conjectured matching \matchingb\ satisfies \autoref{definition:ds-static-stability} in period $t$ and that it induces a dynamically stable matching from period $t+1$ onward says nothing about (a) the existence of blocking pairs in period $t$ involving at least one agent who remains unmatched in period $t$, or (b) whether agents who match in period $t$ under \matchingb\ would prefer to remain unmatched. Both of these decisions depend on the agents other than \agent\ understanding what the \emph{new} continuation matching is, i.e., $(\matchingb(\economyr^{t+\dateindex}))_{\dateindex=1}^{\terminal-t}$, and making decisions optimally relative to this. The optimality of these decisions depends on what matching they conjecture would ensue if they did not follow the prescription of \matchingb. This, in turn, requires knowing what the set of agents available to match in period $t$ is, which presumably does not include \agent. In this case, it is \emph{as if} the agents other than agent \agent\ learn that \agent\ is unavailable to match \emph{before} making the same decision themselves, contradicting the assumption that these decisions are made simultaneously.

Since agents make their blocking decisions simultaneously, the set $M_\ds(\agent,\matching,\economyrt)$ does not describe how \agent\ expects that the remaining agents (and the continuation economy) will \emph{react} to \agent's decision to be unavailable to match at \economyrt. Instead, it describes that even if \agent\ chooses to remain unmatched at \economyrt\ \emph{some} matching $\matchingb(\economyrt)$ will happen at \economyrt. Since dynamic stability describes the set of predictions for the continuation economy, it is possible to pin down the properties of the matching that would ensue in the continuation economy: it will be dynamically stable for $G_{\terminal-t}(\matchingb^t)$. However, it is not immediate to prescribe exactly what $\matchingb(\economyrt)$ should be. On the one hand,
it could be that agents other than \agent\ also decide to be unavailable to match or form blocking pairs.  On the other hand, even if \agent\ is the only one to find fault with the proposed matching, the remaining period $t$ agents are evaluating their decisions to be available to match and/or form blocking pairs using \matching\ as the prescribed outcome and under the assumption that all agents who can match at \economyrt, including \agent, are available to match at \economyrt. It is natural to assume that $\matchingb(\economyrt)$ at the very least satisfies \autoref{definition:ds-static-stability}: When making their decisions to match at \economyrt, no agent should exit with a partner worse than remaining single forever and no two agents who are exiting should have preferred to exit with each other. 

There are two takeaways from the previous discussion. First, the discussion highlights the externality that sets aside dynamic matching markets from their static counterparts. While in a static matching market one should potentially describe the final outcome when there is a block, this is not needed to define the payoffs for the blocking agents. Instead, the dynamic economy forces the analyst to explicitly describe the outcome in period $t$, even if there is a block: otherwise, it is not even possible to define the set of feasible continuation matchings. Second, a potential avenue for further research is to consider refinements of dynamic stability by strengthening \autoref{definition:ds-static-stability}. For instance, one could require that the conjectured period-$t$ matching $\matchingb(\economyrt)$ satisfies that no agent exits with a partner worse than the value of remaining unmatched in period $t$ as described by the set $M_\ds(\cdot,\matching,\economyrt)$.\footnote{For instance, if \matchingbt\ specifies that $\matchingb(\economyrt)(\agenta)\neq\agenta$, then $\agenta$ prefers $\matchingb(\economyrt)(\agenta)$ to the worst element of $M_\ds(\agenta,\cdot)$.} This would be consistent with the idea that agents in period $t$ make their matching decisions using their set of conjectures. What the appropriate strengthening of \autoref{definition:ds-static-stability} is may depend on the application at hand. \autoref{theorem:ds-implementation} in \autoref{sec:sequential} suggests that in sequential assignment problems \autoref{definition:ds-static-stability} is the weakest condition that one can require of $\matchingb(\economyrt)$.

Together with the set of conjectures, there are two other elements of \autoref{definition:ds} worth noting. 

First, \autoref{definition:ds} implies that in order to prevent an agent from blocking matching \matching, it suffices to find a conjectured matching that is worse than \matching. As I explain after the statement of \autoref{theorem:stoch-existence}, the ability to \emph{select} conjectures to dissuade agents from objecting to a matching is important for the existence result. Alternatively, one could require that in order to prevent an agent from blocking matching \matching, \emph{all} conjectured matchings must be worse than \matching. Under this alternative definition, there are economies for which no matching is self-enforcing. To illustrate, consider \autoref{example:main-example}. Under the alternative definition, \erdos\ blocks $\matching^L$ since $\matchingb_{E1}$ is a valid conjecture for \erdos, and \erdos\ prefers $\matchingb_{E1}$ to $\matching^L$
Similarly, under the alternative definition, \kuhn\ blocks $\matchingb^L$: a valid conjecture for \kuhn\ involves \erdos\ and \renyi\ matching in \periodone, in which case, the only stable matching in \periodtwo\ matches \kuhn\ with \shapley. Thus, in the economy in \autoref{example:main-example}, there is no dynamically stable matching under the alternative definition of blocking.


The existing literature offers no general prescription between the two notions of blocking described above (see \cite{ray1997equilibrium} for a similar discussion). However, considering the alternative notion of blocking would potentially require changing other assumptions in the model, in particular, the timing of decisions in each period. 
 To see this, consider an agent \agent\ who chooses to be unavailable to match at \economyrt. A typical argument in favor of the alternative notion of blocking uses \emph{forward induction}: when the agents in period $t+1$ observe that \agent\ deviated from the prescribed matching \matching, they should infer that this only makes sense if there is a continuation matching that \agent\ preferred to matching \matching. However, to select \agent's most preferred conjecture, the forward induction reasoning should apply both to the period-$t$ matching \emph{and} to the continuation matching. That is, the remaining period-$t$ agents should be able to \emph{react} at \economyrt\ to \agent's decision to block. This is inconsistent with the assumption that agents at \economyrt\ make decisions simultaneously. 

The second element is that \autoref{definition:ds} does not require that agents hold ``common beliefs'' about the matching that would ensue when they decide to remain unmatched at \economyrt. That is, suppose two agents \agent\ and \agentp\ anticipate that the matchings $\matchingb_\agent$ and $\matchingb_{\agentp}$ would ensue when they choose to be unavailable to match at \economyrt.  Then, $\matchingb_\agent$ and $\matchingb_{\agentp}$ need not coincide for the agents other than \agent\ and \agentp\ at \economyrt. Furthermore, even if $\matchingb_\agent$ and $\matchingb_{\agentp}$ coincide at \economyrt, they need not coincide in the continuation economy. This property is shared with the model of \cite{sasaki1996two} and also with the models of \cite{ambrus2006coalitional} and \cite{liu2014stable}, where the solution concepts are akin to rationalizability. Note, however, that this lack of commonality only applies to those agents who match in period $t$ and for whom remaining unmatched in period $t$ is an \emph{off path} event. Indeed, the property in \autoref{rem:consistency} implies that in a dynamically stable matching \matching, the agents who remain unmatched at \economyrt\ under \matching\ all share the same conjecture: even if they chose to object to \matching, \matching\ would still be the outcome from period $t$ onward.
\section{Properties}\label{sec:main}
\autoref{theorem:stoch-existence} presents the main result of the paper: the set of dynamically stable matchings is non-empty.
\begin{theorem}\label{theorem:stoch-existence}
For all $\terminal\in\mathbb{N}$, the correspondence $\ds_\terminal$ is non-empty valued.
\end{theorem}
The proof of \autoref{theorem:stoch-existence} is in \autoref{appendix:proof-existence}. The proof shows how to find a matching (labeled \matchingstar\ in the proof) that is dynamically stable. Because of the recursive nature of dynamic stability, the proof needs to simultaneously determine the conjectures, $M_\ds(\agent,\matchingstar,\economyrt)$, and the matching, \matchingstar, that is dynamically stable given these conjectures.

To determine the conjectures, the proof proceeds by induction on $\terminal\geq1$: to show that dynamic stability is well defined for \terminal, one must show that it is well defined for $\terminal^\prime<\terminal$. This is the step that uses the assumption that $\terminal<\infty$. As argued in \autoref{sec:dynamic-stability}, dynamic stability coincides with the static notion of stability when $\terminal=1$. It follows from \cite{gale1962college} that $\ds_1(\cdot)$ is non-empty for all one-period economies. 

Given the set of conjectures, I adapt the proof technique in \cite{sasaki1996two} to find a matching that is dynamically stable. \cite{sasaki1996two} show how to use a set of conjectures to construct an economy without externalities. Building on their insight, I use the agents' conjectures to construct a ``static'' economy in period $t$ as follows. For concreteness, consider the case in which $\terminal=2$. Fix a realization $\economyr_1$. For each agent $k\in A_1\cup B_1$, I calculate \agent's\  continuation value to be the payoff from the worst matching in $M_\ds(\agent,\cdot,\economyr_1)$. When $\terminal=2$, the latter set is non-empty because the set of static stable matchings in period 2 is non-empty regardless of the realization of the arrivals.
With these continuation values at hand, I truncate $k$'s preferences so that $k$ is only willing to match with period-$1$ agents that are preferred to this continuation value. The matching at $\economyr_1$, $\matchingstar(\economyr_1)$, is constructed by running deferred acceptance with agents in $A_1$ proposing to agents in $B_1$ with the truncated preferences. This, in turn, determines the set of unmatched agents at the end of period $1$. For each $\economyr^2$, one then chooses a stable matching amongst the newly arriving agents and the remaining ones from period $1$.

By construction, \matchingstar\ satisfies conditions \ref{itm:ds2} and \ref{itm:ds3} of \autoref{definition:ds}. In particular, for agents who do not match in \periodone\ under \matchingstar, \matchingstar\ is a valid conjecture: the period $2$ matching is an element of $\ds_1$, while the period $1$ matching satisfies \autoref{definition:ds-static-stability}. It only remains to check that \matchingstar\ satisfies condition \ref{itm:ds1} in \autoref{definition:ds}. Suppose there is a pair \pair\ available to match in \periodone, who prefer matching together to \matchingstar. By construction of \matchingstar, it must be that at least one of the members of the pair is unmatched in \periodone. For concreteness, say it is \agenta. It then follows that \agenta\ truncated their preferences ``too much:'' \matchingstar\ is worse than the worst matching in $M_\ds(\agenta,\cdot)$. However, this contradicts that $\matchingstar$ is an element of $M_\ds(\agenta,\cdot)$.

Both the ability to select conjectures to dissuade agents from blocking a matching together with the property in \autoref{rem:consistency}
%
%
are key to show that dynamically stable matchings exist. To see the role of the former, suppose instead that the existence of a conjecture \matchingb\ that is preferred to \matching\ were enough for a single agent to block matching \matching. Then, one could modify the construction in the proof as follows. Instead of truncating an agent's preferences using the worst element in $M_\ds(\cdot)$, one would truncate an agent's preferences using the best element in that set. The same construction would again lead to a matching that satisfies the property in \autoref{rem:consistency} for those agents who remain unmatched in \periodone. However, this is not enough to argue that the unmatched agents would not form a blocking pair with another agent in \periodone. After all, the constructed matching need not be the best conjecture for all the unmatched agents. Similarly, without the property in \autoref{rem:consistency}, the notion of blocking in \autoref{definition:ds} is not enough for the result in \autoref{theorem:stoch-existence} to hold. The ability to promise a blocking agent the worst element in $M_\ds(\cdot)$ maximally dissuades agents who match at \economyrt\ from being unavailable to match. However, there may be agents for whom remaining unmatched at \economyrt\ is the best they can do given this promise. The property in \autoref{rem:consistency} implies that one can effectively deliver this promise simultaneously to all the agents who remain unmatched at \economyrt\ when choosing a dynamically stable matching from period $t+1$ onward.\footnote{This is where the lack of commonality in the conjectures could have actually hindered existence by making it difficult to find one continuation matching that ``works'' for all agents who remain unmatched at \economyrt.} This is because all agents agree that (i) (static) stability is a minimal requirement of the outcome at \economyrt\ and (ii) that continuation matchings are dynamically stable for the continuation economy. 

In static matching markets, the Lone Wolf Theorem (\citealp{mcvitie1970stable}) and the lattice property are important structural properties. However, the set of dynamically stable matchings inherits neither, even if there is no uncertainty over the arrivals:
\begin{prop}\label{prop:lattice}
Suppose the support of \Gterminal\ is a singleton. There exist economies for which the set $\ds_\terminal(\Gterminal)$ does not satisfy the Lone Wolf theorem. 
Furthermore, the set $\ds_\terminal(\Gterminal)$ does not form a lattice.
\end{prop}
The result in \autoref{prop:lattice} is based on an example in \autoref{appendix:lattice}. Despite the structural differences between dynamic stability and stability for static matching markets, there are also important similarities. Like stability in static matching markets, dynamic stability of a matching is a necessary condition for voluntary participation. This is the topic of \autoref{sec:timing}.

\section{Participation, timing, and incentives in dynamic matching markets}\label{sec:timing}
In static matching markets, mechanisms that output stable matchings simultaneously address two problems. First, as observed by \cite{roth1984evolution}, stability is a necessary condition for voluntary participation in the
mechanism. This identifies mechanisms that output stable matchings for the reported preferences as the only candidates to induce participation.
Second, conditional on participating, the mechanism must be such that participants have an incentive
to truthfully report their preferences. It is well-known that no stable mechanism exists that makes it optimal for both sides of the
market to truthfully reveal their preferences. Instead, if the mechanism outputs, say, the matching obtained by side $\arrivala$-proposing
deferred acceptance (henceforth, DA), then it is a dominant strategy for side $\arrivala$ to truthfully report their preferences. Thus, in settings where side 
$\arrivalb$ is non-strategic, side $\arrivala$-proposing DA is used to simultaneously address the problems of participation
 and incentives.
 
This section analyzes the issues of participation and incentives in the context of dynamic matching markets. \autoref{prop:stoch-queueing} in \autoref{sec:voluntary} shows that dynamic stability is a necessary condition for \emph{timely} participation in the matching market. Furthermore, an extension of the DA algorithm to the dynamic economy is shown to give agents on the proposing side the correct incentives to participate as soon as they arrive. \autoref{sec:sequential} studies the issues of participation and incentives in the context of sequential assignment problems. The main result in \autoref{sec:sequential} (\autoref{theorem:ds-implementation}) shows that only dynamically stable matchings can arise as equilibrium outcomes of the preference revelation game induced by a sequence of stable matching algorithms.
\subsection{Voluntary and timely participation}\label{sec:voluntary}
\citet[pp.~22--23]{roth1992two} argue that stability is a necessary condition for voluntary participation in a static economy as follows. Suppose a matchmaker can recommend a matching for the economy, but cannot compel the agents to accept the matching. Instead, the agents are free to form pairs amongst themselves or choose to remain unmatched. Then, the agents follow the matchmaker's recommended matching only if it is a stable matching. \autoref{prop:stoch-queueing} shows that dynamic stability plays the same role in the dynamic economy: the agents follow the matchmaker's recommended matching only if it is dynamically stable.\footnote{The mechanism design literature on revenue management studies the problem of agents strategically reporting their availability to participate in the mechanism (e.g., \cite{gershkov2015efficient,garrett2016intertemporal,bergemann2019progressive}).}
%

Consider a matchmaker who recommends a matching \matching\ for the agents in \Gterminal\ and whose objective is that the agents follow her recommendation.  In each period $t$, 
agents arrive according to \Gterminal. Upon arrival, each agent observes the agents who have arrived through period $t$ and whether they have already matched. Upon arrival, each agent chooses whether to reveal that they have arrived. Given a sequence of reported arrivals through period $t$, \economyrbt, agents are matched according to $\matching(\economyrbt)$.

In what follows, I explain the assumptions needed to state \autoref{prop:stoch-queueing}. The result requires that the distribution over \economyterminal\ satisfies an independence condition. Say $\Gterminal$ satisfies \emph{exchangeability} if $\Gterminal(\economyr_1,\dots,\economyr_\terminal)=\Gterminal(\tilde{\economyr}_1,\dots,\tilde{\economyr}_\terminal)$ whenever $\arrivalauptoterminal(\economyrterminal)=\arrivalauptoterminal(\tilde{\economyr}^\terminal)$ and $\arrivalbuptoterminal(\economyrterminal)=\arrivalbuptoterminal(\tilde{\economyr}^\terminal)$. That is, the probability of two realizations \economyrterminal\ and $\tilde{\economyr}^\terminal$ only depends on the arrivals they induce through period \terminal\ (and not on how they came to be.)
Finally, say that a matching \matching\ for economy \Gterminal\ is not dynamically stable for arriving agents if whenever there exists a realization \economyrt\ and an agent $\agent\in\cala(\matchinguptot,\economyrt)\cup\calb(\matchinguptot,\economyrt)$  who would benefit from waiting to be matched, then there exists an agent $\agent^\prime$ who would benefit from waiting to be matched and arrived in period $t$.
We are now ready to state \autoref{prop:stoch-queueing}:
\begin{prop}\label{prop:stoch-queueing}
Suppose \Gterminal\ has full support and satisfies exchangeability. Let $\matching\in\Matchingsfinal$ be pairwise stable and individually rational and suppose that \matching\ is not dynamically stable for arriving agents.  Then, there exists a period $t$, a realization $\economyrt=(E^{t-1},A_t,B_t)$, and $k\in A_t\cup B_t$ such that
\begin{enumerate}[leftmargin=*]
\item $\matching(\economyrt)(\agent)\neq \agent$, and 
\item $\agent$ prefers $\matching(\economyrt\setminus\{\agent\},\cdot)$ to $\matching(\economyrt)$. Formally, if $\agent=\agenta$, then $\economyrt\setminus\{\agenta\}=(\economyr^{t-1},A_t\setminus\{\agenta\},B_t)$ and 
\small
\begin{align}\label{eq:regreta}
u(\agenta,\matching(\economyrt)(a))<\mathbb{E}_{G(\cdot|\economyrt)}[\delta_\agenta^{t_\matching(\economyrt\setminus\{\agenta\},\economyr_{t+1}\cup\{\agenta\},\cdot)}u(\agenta,\matching(\economyrt\setminus\{\agenta\},\economyr_{t+1}\cup\{\agenta\},\cdot)(\agenta))],
\end{align}\normalsize
and similarly if $k=b$, then $\economyrt\setminus\{\agentb\}=(\economyr^{t-1},A_t,B_t\setminus\{\agentb\})$ and 
\small
\begin{align}\label{eq:regretb}
v(\matching(\economyrt)(b),b)<\mathbb{E}_{G(\cdot|\economyrt)}[\delta_b^{t_\matching(\economyrt\setminus\{b\},\economyr_{t+1}\cup\{b\},\cdot)}v(\matching(\economyrt\setminus\{\agentb\},\economyr_{t+1}\cup\{b\},\cdot)(b),b)].
\end{align}\normalsize
\end{enumerate}
\end{prop}
To fix ideas, suppose $k=a$. Then, \autoref{eq:regreta} shows that $a$ would improve on his outcome by waiting until period $t+1$ to report that they have arrived.\footnote{This does not imply, however, that delaying by one period is the best that $a$ can do.} That is, when \matching\ is not dynamically stable, $a$ has an incentive to lie about their availability to match in period $t$. 
%
When everyone else reports their arrivals truthfully,  \agenta\ induces in period $t$ the matching $\matching(\economyrt\setminus\{\agenta\})$ by lying about their arrival; next period, when \agenta\ joins the economy, the matching is then $\matching(\economyrt\setminus\{\agenta\},\economyr_{t+1}\cup\{\agenta\},\cdot)$ (note that \agenta\ only assigns positive probability to realizations $\economyr_{t+1}$ such that $\agenta\notin\economyr_{t+1}$.) 

 When $\terminal=2$, \autoref{prop:stoch-queueing} follows straight from the definition of dynamic stability. When $\terminal\geq 3$, however, this is not the case. The issue is that \matching\ may fail to be dynamically stable for arriving agents at \economyrt,\ and yet \agenta\ reveals his arrival truthfully because $\matching(\economyrt\setminus\{\agenta\},\cdot)$ does not entail a continuation matching that is dynamically stable.\footnote{Since \matching\ is pairwise stable and individually rational, $\matching(\economyrt\setminus\{a\})$ is stable amongst those who match in period $t$.} In this case, $\matching(\economyrt\setminus\{\agenta\},\cdot)$ could be worse than any \emph{reasonable} conjecture \agenta\ may have under the assumption of dynamic stability. The key is then to find the longest realization \economyrt\ for which dynamic stability fails. This ensures that $\matching(\economyrt\setminus\{\agenta\},\cdot)$ does pick continuation matchings that satisfy dynamic stability for $\Gterminal(\cdot|\economyrt\setminus\{a\},\economyr_{t+1}\cup\{a\})$ from period $t+1$ onward. The assumption on $\Gterminal$ ensures that these continuations are also dynamically stable when using $\Gterminal(\cdot|\economyrt,\economyr_{t+1})$.

\autoref{prop:stoch-queueing} highlights that in a dynamic economy voluntary participation involves not only the decision of whether to participate, but also \emph{when}. \autoref{prop:ds-sp} and \autoref{ex:college-admissions} below further illustrate this point. While in static matching markets DA provides the agents on the proposing side with the correct incentives to report their preferences, \autoref{prop:ds-sp} shows that when $\terminal=2$, a natural extension of DA to the dynamic economy provides the agents on the proposing side with the correct incentives to participate as soon as they arrive. 
%

To state \autoref{prop:ds-sp}, let \mada\ denote the matching obtained by running the following dynamic version of DA: the DA algorithm is run with agents in $\arrivala_1\cup\arrivala_2$ making proposals (using their intertemporal preferences) to agents in 
$\arrivalb_1\cup\arrivalb_2$ (who choose between proposals using their intertemporal preferences). The following holds:
\begin{prop}\label{prop:ds-sp}
Let $\terminal=2$ and let $\Gtwo$ be as above. Then, for all $a\in A_1$, there exists $\matchingb\in M_\ds(a,\mada,\economyr^1)$ such that 
\[U(\agenta,\mada,\economyr^1)\geq U(\agenta,\matchingb,\economyr^1).\]
\end{prop}

\autoref{prop:ds-sp} states that, in a two-period economy, agents in $\arrivala_1$ cannot improve on the outcome of \mada\ by waiting to be matched. To see this, fix an agent $\agenta\in\arrivala_1$. Consider the matching $\matchingb^{A-DA}$ obtained by running the dynamic version of side-\arrivala\ DA in the economy $(\arrivala_1\setminus\{\agenta\},\arrivalb_1,\arrivala_2\cup\{\agenta\},\arrivalb_2)$. That is, agent \agenta\ makes proposals \emph{as if} \agenta\ arrived in \periodtwo, and agents $b\in\arrivalb_1$ evaluate the payoff from matching with agent \agenta\ \emph{as if} \agenta\ arrived in \periodtwo. I show that $\matchingb^{A-DA}$ is an element of $M_\ds(\agenta,\cdot)$. If \agenta\ could improve on \mada\ by waiting to be matched, then \agenta\ strictly prefers $\matchingb^{A-DA}$ to \mada. Since DA is strategy-proof in static matching markets,
it cannot be that \agenta\ prefers $\matchingb^{A-DA}$ to \mada. Thus, $\matchingb^{A-DA}$ can be chosen as the matching \agenta\ expects would arise if \agenta\ blocked matching \mada\ in \periodone.

Instead, the matching obtained by running DA  with agents in $\arrivalb_1\cup\arrivalb_2$ making proposals (using their intertemporal preferences) to agents in $\arrivala_1\cup\arrivala_2$ may be improved upon agents on side $\arrivala$ waiting to be matched, as the next example shows:
\begin{ex}\label{ex:college-admissions} 
Consider the following variant of \autoref{example:main-example}. \kuhn\ and \gale\ continue to arrive at \periodone,\ while \erdos\ now arrives at \periodtwo.\ Arrivals on side $\arrivalb$ are as before, except that now \renyi\ no longer arrives. That is, $\arrivala_1=\{\kuhn,\gale\}$, $\arrivalb_1=\{\tucker\}$, $\arrivala_2=\{\erdos\}$ and $\arrivalb_2=\{\shapley,\nash\}$. Preferences are given by:
\small
\[\hspace{-1cm}\left.\begin{array}{lllll}
\erdos\ :&\nash\ &&&\\
\kuhn\ :&(\shapley,1)&(\nash,0)&(\tucker,0)&(\nash,1)\\
\gale\ :&(\tucker,0)&(\shapley,1)&&
\end{array}
\right.\hspace{0.25cm}
\left.\begin{array}{llll}
\tucker: &\kuhn\ &\gale\ &\\\shapley:&\gale\ &\kuhn\ &\\
\nash:&\kuhn\ &\erdos&
\end{array}
\right.\]
\normalsize

\autoref{fig:ex-addendum} illustrates the matchings obtained by running  DA with side $\arrivalb$ and side $\arrivala$ proposing, respectively:

\begin{figure}[h!]\vspace{-0.25cm}
\[\mbda=\left(\begin{array}{lcl}
\kuhn\ &\_\_&\tucker\ \\
\hline
\gale\ &\_\_&\shapley\ \\
\erdos\ &\_\_& \nash\
\end{array}\right)\hspace{1.5cm}\mada=\left(\begin{array}{lcl}
\gale\ &\_\_&\tucker\ \\
\hline
\kuhn\ &\_\_&\shapley\ \\
\erdos\ &\_\_&\nash\
\end{array}\right)\]\vspace{-0.25cm}
\caption{Two matchings; the one on the left is not dynamically stable}\label{fig:ex-addendum}
\end{figure}

The matching \mbda\ is not dynamically stable: \kuhn\ can guarantee to be matched with \shapley\ by remaining unmatched in \periodone. To see this, note that \gale\ also needs to match in \periodtwo\ in order for \kuhn\ not to match with \shapley\ in \periodtwo. Hence, \kuhn\ needs to conjecture that \emph{everyone} matches in \periodtwo\ when he waits to be matched. However, when this is the case, all stable matchings match \kuhn\ with \shapley. Thus, if the matchmaker were to suggest \mbda, \kuhn\ would not follow the matchmaker's recommendation in \periodone.
%
\end{ex}
\autoref{prop:ds-sp} and \autoref{ex:college-admissions} echo the static matching markets results that side $\arrivala$-proposing DA is strategy-proof for agents on side $\arrivala$, while side $\arrivalb$-proposing DA is not. Indeed, the proof of \autoref{prop:ds-sp} is intimately related to the strategy-proofness of the DA algorithm for the proposing side. 
However, the analogy between the result in \autoref{prop:ds-sp} and the strategy-proofness of deferred acceptance for the proposing side is incomplete. While the latter refers to agents' incentives to truthfully report their preferences \emph{conditional} on participating in the mechanism, the former refers to agents' incentives to timely participate in the mechanism when their preferences are commonly known. In other words, in a dynamic matching market agents can strategically decide when to participate and what preferences to report. This is the focus of \autoref{sec:sequential}.
\subsection{Sequential assignment problems}\label{sec:sequential}
\autoref{sec:sequential} studies the issues of timely participation and preference manipulation in dynamic matching markets within the context of sequential assignment problems. In sequential assignment, matchings are performed in multiple stages via a sequence of \emph{spot} mechanisms, which take the current set of available agents and their rankings and outputs a matching as a function of the reported preferences, but do not condition on future matching opportunities. 

Sequential assignment covers important applications like school choice and college admissions, where both students and school seats become available over time. In school choice, admissions to public, charter, and private schools often occur at different times and through different mechanisms, schools update the number of seats available as the beginning of the school year approaches, and new students join the public school system during the summer as families move across district and/or state boundaries. In many districts in the US, this leads to \emph{aftermarkets} (\citealp{pathak2016really}):
Public school districts run their matching algorithms several times to accommodate newly incoming students, newly available seats, and also the timing of decisions of private and charter schools. While in the US private schools do not participate in the centralized matching procedure, they do in Turkey and some localities in Sweden.\footnote{\cite{andersson2018sequential} provide an excellent description of different school districts (Boston, New York, Turkey, Sweden, amongst others) and their sequential algorithms.} 
For instance, seats in public and private schools are assigned via a two-stage procedure in Turkey. Since 2015, this two-stage procedure operates as follows: In the first stage, students are assigned based on test scores via serial dictatorship to private schools. In the second stage, unmatched students from the first stage are assigned based on test scores via serial dictatorship to public schools.\footnote{Before 2015, public school seats were assigned first and even if they received a match in the first round, students could participate in the second round, where private school seats were assigned.} 
Another example comes from college admissions in Germany ( \citealp{westkamp2013analysis}): In the first stage students with high grades and/or high wait times are assigned through the Boston mechanism to a subset of the available seats, and then the remaining students and seats are matched through college-proposing DA.

The use of spot mechanisms presents a well-known problem in dynamic environments: implementing an allocation that has good dynamic properties may require using information beyond that which is available in the current period (\citealp{parkes2007online}). \autoref{prop:stoch-queueing} identifies dynamically stable matchings as those that induce the right incentives to participate. As \autoref{ex:college-admissions-bis} below illustrates, spot mechanisms do not necessarily implement dynamically stable matchings:
\begin{ex}\label{ex:college-admissions-bis}
Consider again the economy in \autoref{ex:college-admissions}. 
I reproduce below the matchings obtained using the dynamic version of DA:

\begin{figure}[h!]\vspace{-0.25cm}
\[\mbda=\left(\begin{array}{lcl}
\kuhn\ &\_\_&\tucker\ \\
\hline
\gale\ &\_\_&\shapley\ \\
\erdos\ &\_\_& \nash\
\end{array}\right)\hspace{1.5cm}\mada=\left(\begin{array}{lcl}
\gale\ &\_\_&\tucker\ \\
\hline
\kuhn\ &\_\_&\shapley\ \\
\erdos\ &\_\_&\text{Nash}
\end{array}\right)\]\vspace{-0.25cm}
\caption{The matchings in \autoref{fig:ex-addendum}}\label{fig:ex-spot}
\end{figure}

Suppose instead that one runs DA (with either side proposing) among the agents in \periodone, and then one runs DA among the remaining unmatched agents and the new arrivals in \periodtwo. In this example, the resulting matching would be \mbda.
%
%
Furthermore, in this example, \mbda\ would also be the outcome of running the Boston mechanism in \periodone, followed by side-\arrivalb\ DA in \periodtwo, as in German college admissions. Since \mbda\ is not dynamically stable, it follows that neither of these sequences of spot mechanisms  can 
produce a dynamically stable matching for this economy.

Instead, the analysis in \autoref{ex:college-admissions} implies that the matching \mada\ is dynamically stable and can be achieved via the dynamic version of side-\arrivala\ DA. To understand why \mada\ cannot be achieved using either sequence of spot mechanisms in the previous paragraph, note the following.
%
 The period-$1$ matching under \mada\ matches \gale\ with \tucker, leaving \kuhn\ unmatched. This poses no challenges to stability in the dynamic economy: \kuhn\ prefers to remain unmatched in \periodone\ so as to match with \shapley\ in \periodtwo. However, this matching is not stable in \periodone\ relative to the agents' true preferences: \tucker\ prefers \kuhn\ over \gale, and \kuhn\ prefers \tucker\ over remaining unmatched. 

The results in \autoref{sec:voluntary} imply that if \mbda\ is the matching that is to be implemented by the sequence of spot mechanisms, \kuhn\ will not find it optimal to participate in \periodone. In sequential assignment problems, \kuhn\ not only chooses when to participate, but also what preferences to report. In this case, instead of not participating in the first stage, \kuhn\ could submit a ranking that only lists \shapley\ in \periodone. Doing so guarantees that he will remain unmatched in \periodone, and be matched with \shapley\ in \periodtwo.
That is, even if the period-$1$ matching is determined by side-\arrivala\ DA, \kuhn\ is better off by misreporting his preferences.

\kuhn's deviation in the previous paragraph is consistent with the recommendation received by German students in the college admissions procedure: they should 
truncate their preferences in the first round if they wish to be considered for the next round (see \citealp{westkamp2013analysis}). Indeed, \cite{braun2010telling} report that students with high grades truncate their preferences substantially since given their grades they have good chances in the second stage, where there is more choice. As the example illustrates, agents may benefit from truncating their preferences even if a stable and strategy-proof mechanism is used.
\end{ex}
Two lessons follow from \autoref{ex:college-admissions-bis}. First, existing mechanisms used in sequential assignment problems fail to deliver dynamically stable matchings. Second, whenever this is the case, agents' incentives either to participate or to truthfully report their preferences might be hindered. 

\autoref{theorem:ds-implementation} shows that agents' forward-looking behavior is enough to overcome the inability of spot mechanisms to produce dynamically stable matchings. Indeed, \autoref{theorem:ds-implementation} shows that only dynamically stable matchings can arise as the outcomes of pure strategy SPNE of the game induced by a sequence of spot mechanisms that implement stable matchings. However, as the proof of \autoref{theorem:ds-implementation} and \autoref{ex:college-admissions-bis} above illustrate, achieving dynamically stable matchings maybe at odds with truthful behavior: Agents may truncate their preferences so that they are never matched to a matching partner that is worse than what they would obtain by waiting to be matched.

To state \autoref{theorem:ds-implementation}, I formally define a spot matching mechanism and, given a sequence of spot mechanisms, the non-cooperative game for \Gterminal\ induced by this sequence. 

A spot stable mechanism, denoted in what follows by \spot, is a mapping that takes two sets of agents, one on each side, and their reported preferences and outputs a matching that is stable given the reported preferences. Formally, for an agent $\agenta$ a rank ordered list (henceforth, ROL) is a ranking over $\sideb\cup\{\agenta\}$. Similarly, for an agent \agentb\ a ROL is a ranking over $\sidea\cup\{\agentb\}$. Let $\rol_\agent$ denote agent \agent's ROL. Thus, a spot mechanism, \spot, takes a tuple $(\arrivalab,\arrivalbb,\overline{\succ})$ and outputs $\spot(\arrivalab,\arrivalbb,\profilerol)\in\stable(\arrivalab,\arrivalbb,\profilerol)$, where $\arrivalab\subseteq\sidea,\arrivalbb\subseteq\sideb$ and $\profilerol=(\rol_\agent)_{\agent\in\arrivalab\cup\arrivalbb}$ is a profile of ROLs. Note that, in a slight abuse of notation, I index the set of stable matchings, \stable, both by the set of agents and their reported preferences.

A sequence of spot mechanisms, $\{\spot_t\}_{t=1}^\terminal$ induces the following extensive form game. In each period $t$, the remaining unmatched agents and the new arriving agents observe who has matched through period $t$. They decide simultaneously whether to participate in the mechanism, and if so, what ROLs to submit. The decision to participate and the submitted ROLs are not observable. Given the set of participants and their ROLs, the spot mechanism $\spot_t$ outputs a matching.  Matched agents and their partners exit.\footnote{Thus, the game mimics how the mechanisms in the National Resident Matching Program, college admissions in Germany, school choice in Turkey, and in some localities in Sweden operate: only unmatched agents participate in the upcoming rounds. Prior to 2015, when matched agents were allowed to participate in the second round, the government in Turkey had to supplement the algorithm with five more rounds (see \cite{andersson2018sequential}). } The game proceeds to period $t+1$.

\autoref{theorem:ds-implementation} considers two variants of this game. In the first variant, denoted by \gameda, the spot mechanisms coincide with side $\arrivala$ DA. Moreover, the agents on side $\arrivalb$ are non-strategic: they automatically join when they arrive and submit their true rankings. \autoref{theorem:ds-implementation} studies which matchings can result as the outcome of SPNE. They satisfy a property denoted \emph{side-$\arrivala$ dynamic stability}: the matching must be pairwise stable and no agent on side \arrivala\ finds it optimal to wait to be matched.\footnote{Formally, side-$A$ dynamic stability is the solution concept defined recursively by conditions \ref{itm:ds2} and \ref{itm:ds1} in \autoref{definition:ds}. Note that this implies that when agents on side $\arrivala$ consider waiting to be matched, they anticipate that the continuation matching is side $A$-dynamically stable.} In the second variant, denoted by \gamestable, the spot mechanisms can be any stable matching algorithm and both sides are strategic. \autoref{theorem:ds-implementation} studies which matchings can result as the outcome of a refinement of SPNE, \emph{pairwise SPNE}, which allows pairs of agents who are present at the same time to jointly deviate.\footnote{
This refinement is common in the literature that studies the non-cooperative implementation of stable matchings  (see, for instance, \cite{ma1995stable,shin1996mechanism,sonmez1997games}) to address an observation in \cite{alcalde1996implementation} that Nash equilibrium is not enough to guarantee that the equilibria of the non-cooperative game induced by a stable matching algorithm coincides with the set of stable matchings.}

We are now ready to state \autoref{theorem:ds-implementation}:
\begin{theorem}\label{theorem:ds-implementation}
Only side-$\arrivala$ dynamically stable matchings for \Gterminal\ can be the outcome of pure strategy SPNE in \gameda. Similarly, only dynamically stable matchings for \Gterminal\ can be the outcome of pure strategy pairwise SPNE in \gamestable. 
 \end{theorem}
 The proof is in \autoref{appendix:sequential-implementation} and it proceeds by backward induction. I highlight here the main insights that come from the proof. 

  First, building on the results in \cite{roth1991incentives}, I show that it is without loss of generality to focus on equilibrium strategies where agents' ROLs are either truncations of the rankings induced by their Bernoulli utility functions or only list themselves, i.e., an \emph{empty} ROL. 
Agents may need to submit empty ROLs if their most preferred matching partner is not available in a given round, so as to avoid being matched to someone that is worse than waiting to be matched.\footnote{Alternatively, an agent can always list their most preferred matching partner, even if they are not available in a given round.} It follows that when agents are allowed to submit different ROLs in each stage (as in the applications described so far), agents can always participate in the mechanism.

Second, I show that for all periods $t$ and for all matchings that may have ensued through period $t$, the outcome of equilibrium play starting from period $t$ is a dynamically stable matching. The properties of deferred acceptance in \gameda\ and the possibility of joint deviations in \gamestable\ imply that the matching satisfies condition \ref{itm:ds1}. To show that no agent benefits from remaining unmatched in a given period, I show that the matching that would result when an agent deviates and submits a ROL that leaves them unmatched is a valid conjecture. Therefore, if the matching were not dynamically stable, the agent would have a deviation. To see why such a deviation induces a matching that is a valid conjecture, fix a period $t$ and an agent \agent\ who can match in period $t$. Suppose one has already shown that from stage $t+1$ onward, equilibrium play leads to a dynamically stable matching. This means that, conditional on remaining unmatched in period $t$, \agent\ expects that the matching starting from period $t+1$ satisfies condition (ii) in \autoref{eq:conjectures}. Because agents may not report their preferences truthfully, the period-$t$ matching that results when \agent\ deviates is not necessarily stable with respect to the true preferences. However, the period-$t$ matching induced by \agent's deviation always satisfies \autoref{definition:ds-static-stability}: amongst the agents that do match in period $t$, the presence of a blocking pair would contradict that agents submit truncations of their true preference ranking. 

  \autoref{theorem:ds-implementation} complements the results in \cite{westkamp2013analysis} and \cite{dur2019sequential}. \cite{westkamp2013analysis} shows that the mechanism used in German college admissions fails to produce (static) stable matchings and proposes a one-shot mechanism that respects the priority of those students with either high-grades and high-waiting times. In a two-period model, \cite{dur2019sequential} show that sequential assignment may be at odds with (static) stability and/or truthful behavior even when using stable and strategy-proof mechanisms in each round. They also show which spot mechanisms are such that the outcome of pure strategy Nash equilibria is a stable matching.

  \autoref{theorem:ds-implementation} contributes to their analysis by identifying the stability property satisfied by the outcomes of equilibrium play \emph{regardless} of the sequence of stable mechanisms used in each round. Furthermore, it illustrates the form that manipulations take in dynamic environments. Even if the spot mechanism is strategy-proof for one side, one should expect that agents truncate their preferences so that their assignments reflect what the agents expect to be able to guarantee should they stay for an extra round.\footnote{Strategic mistakes, as in \cite{shorrer2018obvious} and \cite{hassidim2020limits}, is another reason why agents may misreport their preferences in strategy-proof mechanisms in static settings.\label{strategic-mistakes}} 
These two observations can inform applied researchers that study sequential assignment problems (e.g., \cite{narita2018match,neilson2020aftermarket}). First, \autoref{theorem:ds-implementation} warns against the use of the reported preferences as the true preferences, even if a strategy-proof mechanism is used.\footnote{\cite{narita2018match} reports manipulations that are not truncations. This is to be expected since the NYC match restricts the number of schools that parents may list, while the sufficiency of truncations is established under the assumption that agents can submit lists of any length.} Second, stability notions are oftentimes used for preference identification and \autoref{theorem:ds-implementation} identifies dynamic stability as the solution concept for these applications.
  
While \autoref{theorem:ds-implementation} identifies dynamic stability as the property satisfied by matchings that arise from equilibrium behavior in sequential assignment problems, it should not be interpreted as an implementation result for dynamically stable matchings.\footnote{Instead, the game in \cite{lagunoff1994simple} can be used to provide an implementation of dynamically stable matchings.} Indeed, a given sequence of spot mechanisms may not implement all dynamically stable matchings in a given economy. Moreover, \autoref{theorem:ds-implementation} does not assert a pure strategy SPNE exists. Instead, \autoref{theorem:ds-implementation} should be interpreted as stating that dynamic stability is a necessary condition: whenever the matching that results from sequential assignment is not dynamically stable, either a pair of agents would prefer to match outside the algorithm, or an agent will find it optimal to delay the time at which they are available to match. 
 \section{Further directions}\label{sec:conclusions}
While stability is a key property in the analysis of static matching markets, the analysis of dynamic matching markets has been confined to either equilibrium models or analyzed through the lens of static notions of stability. This paper fills this gap by formulating a stability notion for dynamic matching markets. As such, the paper opens several avenues for further research. First, as the discussion at the end of \autoref{sec:dynamic-stability} suggests, one could consider refinements of dynamic stability by strengthening condition (iii) in \autoref{eq:conjectures}. Indeed, one such refinement is proposed in that section. Second, developing an algorithm that implements dynamically stable matchings is definitely of interest. When $\terminal=2$, such an algorithm can be built using the proof of \autoref{theorem:stoch-existence}. However, it is an open question whether such an algorithm exists for $\terminal\geq3$. Finally, the analysis in \autoref{sec:sequential} identifies dynamic stability as the solution concept in sequential assignment problems. Since many applications in sequential assignment involve many-to-one matching markets, it would be natural to extend \autoref{definition:ds} to many-to-one markets. \cite{altinok2019dynamic} is a step in this direction. 
%
%
%
%
%

\bibliographystyle{ecta}
\bibliography{matching}

\appendix
\section{Proof of Theorem \ref{theorem:stoch-existence}}\label{appendix:proof-existence}
The proof proceeds by induction on $T\geq 2$. As argued in the main text $\ds_1$ coincides with the set of pairwise stable matchings for the static economy. Let $\mathbf{P(\terminal)}$ denote the following inductive statement:

$\mathbf{P(T):}$ The correspondence $\ds_\terminal$ is non-empty valued. 

I first prove $\mathbf{P(2)}=1$. Let $\Gtwo$ denote an economy of length $\terminal=2$. For each $\economyr^1=(A_1,B_1)$ on the support of $\Gtwo$, I construct a matching $\matchingstar(\economyr^1)$ as follows.
For each $k\in A_1\cup B_1$, let $\matchingb_{\economyr^1}^k$ denote the payoff minimizing matching \matching\ such that
\begin{enumerate}[leftmargin=*]
\item $\matching(\economyr^1)(k)=k$,
\item\label{itm:p22} For each $\economyr^2=(\economyr^1,A_2,B_2)$ in the support of $\Gtwo$, $\matching(\economyr^1,A_2,B_2)$ is stable for $\langle\cala(\matching,\economyr^2),\calb(\matching,\economyr^2)\rangle$,
\item\label{itm:p23} $\matching(\economyr^1)$ is stable amongst those who match at $\economyr^1$.
\end{enumerate}
Note that the set of matchings \matching\ that satisfy these three requirements is finite and it is non-empty. To see that it is non-empty, note first that for any set of unmatched agents, the set of stable matchings in \periodtwo\ is non-empty. Second, the set of matchings that are (statically) stable for the \periodone\ economy that consists of agents in $A_1\cup B_1\setminus\{k\}$ is non-empty. Pick any such \periodone\ matching and extend it to $A_1\cup B_1$ by leaving agent $k$ unmatched. Note that it is stable amongst those who match in \periodone.\ Then, for each $\economyr^2$, pick a \periodtwo\ matching that is stable for $\langle\cala(\matching_1,\economyr^2),\calb(\matching_1,\economyr^2)\rangle$. The matching \matching\ formed this way satisfies the three requirements. 
Also note that \autoref{itm:p22} implies that any such matching \matching\ satisfies the requirement that $\matching(\economyr^1,\cdot)\in\ds_1(\Gtwo(\matching^1,\economyr^1))$. 

For each $\agenta\in\arrivala_1$, consider the preference list on $\arrivalb_1$ defined as follows: $\agentb\in \arrivalb_1$ is acceptable to $\agenta$ only if $\genericu\geq U(\agenta,\matchingb_{\economyr^1}^\agenta,\economyr^1)$ and is not acceptable to $\agenta$ otherwise. Moreover, if $\agentb,\agentbp\in\arrivalb_1$ are acceptable to $\agenta$, then they are ordered according to $\genericuo\cdot)$ in $a$'s preference list.

Similarly, for $\agentb\in\arrivalb_1$, consider the preference list on $\arrivala_1$ defined as follows: $\agenta\in\arrivala_1$ is acceptable to $\agentb$ only if $\genericv\geq V(\agentb,\matchingb_{\economyr^1}^\agentb,\economyr^1)$ and is not acceptable to $\agentb$ otherwise. Moreover, if $\agenta,\agentap\in\arrivala_1$ are acceptable to $\agentb$, then they are ordered according to $v(\cdot,\agentb)$ in $\agentb$'s preference list.

Define $\matchingstar(\economyr^1)$ so that it coincides with the outcome of running the deferred acceptance algorithm with agents in $\arrivala_1$ proposing to agents in $\arrivalb_1$ using the \emph{truncated} preference lists.(If there are ties, fix a tie-breaking procedure and run deferred acceptance.) 

 For each period-2 economy $\economyr^2=(\economyr^1,A_2,B_2)$, let $\matchingstar(\economyr^2)$ coincide with $\matchingstar(\economyr^1)$ for those agents who matched at $\economyr^1$ and have it coincide with the outcome of running deferred acceptance on $\left(\cala(\matchingstaruptoone,\economyrtwo),\calb(\matchingstaruptoone,\economyr^2)\right)$, elsewhere.\footnote{Again, if there are indifferences, run deferred acceptance after having broken ties.}

Note that \matchingstar\ defined in this way is an element of $\Matchings_2$.

I note the following properties of $\matchingstar$. First, for all $\economyr^2$, $\matchingstar(\economyr^2)$ is stable for $\left(\cala(\matchingstaruptoone,\economyr^2),\calb(\matchingstaruptoone,\economyr^2)\right)$. Second, for each $\economyr^1$, \matchingstar\ satisfies the following: (i) there is no pair \pair\ such that $\matchingstar(\economyr^1)(a)\neq a,\matchingstar(\economyr^1)(b)\neq b$ that prefer matching with each other than to match according to $\matchingstar(\economyr^1)$, i.e., $\matchingstar(\economyr^1)$ is stable amongst those who match in period $1$; (ii) there is no $k\in\arrivala_1\cup\arrivalb_1$ such that $\matchingstar(\economyr^1)(k)\neq k$ and for all $\matchingb\in M_{\ds}(k,\matchingstar,\economyr^1)$, $k$ prefers $\matchingb$ to $\matchingstar$. To see that (ii) holds, note the following. By construction $\matchingb^k$ can always be taken to be an element of $M_{\ds}(k,\matchingstar,\economyr^1)$ by having it coincide with $\matchingstar$ at all realizations $\economyrt$, $t\in\{1,2\}$ that do not (weakly) follow $\economyr^1$. Moreover, if $k$ is matched at $\economyr^1$, then it must be that $k$'s matching partner is (weakly) preferred to $\matchingb^k$.
\newline\indent Thus, to show that $\matchingstar\in\ds_2(\Gtwo)$, it remains to be argued that it is pairwise stable. To do so, it is only necessary to check that there is no realization $\economyr^1=(\arrivala_1,\arrivalb_1)$ and pair $(\agenta,\agentb)\in\arrivala_1\times\arrivalb_1$ who prefer to match together at $\economyr^1$ rather than match according to $\matchingstar$. Toward a contradiction, suppose there is such a pair. Since $\matchingstar$ is stable amongst those who match in period $1$, it cannot be the case that both \agenta\ and \agentb\ are matched at $\economyr^1$. Then, either \agenta\ or \agentb\ (or both) declared the other unacceptable at $\economyr^1$ and is unmatched at $\economyr^1$. Without loss of generality, assume that it is \agentb,\ then
\begin{align}\label{eq:aunacceptabletob}
\genericv<V(b,\matching_{\economyr^1}^b,\economyr^1).
\end{align}
Moreover, it must be that $\genericv>V(\agentb,\matchingstar,\economyr^1)$ since (\agenta,\ \agentb)\ form a pairwise block of $\matchingstar$. It then follows that
\begin{align}\label{eq:wronglowerbound}
V(b,\matchingstar,\economyr^1)<V(b,\matchingb_{\economyr^1}^\agentb,\economyr^1).
\end{align}
Note, however, that \matchingstar\ satisfies the following: (i) $\matchingstar(\economyr^1)(\agentb)=\agentb$, (ii) $(\matchingstar(\economyr^1,\cdot))\in\ds_1(\Gone(\matchingstaruptoone,\economyr^1))$, and (iii) $\matchingstar(\economyr^1)$ is stable amongst those who match in period $1$. Then, $\matchingstar\in M_{\ds}(\agentb,\matchingstar,\economyr^1)$. This, together with \autoref{eq:wronglowerbound}, contradicts the definition of $\matchingb_{\economyr^1}^\agentb$. By assuming instead that $a$ is unmatched, a similar contradiction follows. I conclude that \matchingstar\ is pairwise stable. This concludes the proof that $\mathbf{P(2)}=1$.
\newline\indent I now prove the inductive step. Assume now that $\mathbf{P(\terminal^\prime)}=1$ for all $\terminal^\prime<\terminal$. I show that $\mathbf{P(\terminal)}=1$. 

Given a realization $\economyrterminal=(A_1,B_1,\dots,A_\terminal,B_\terminal)$ in the support of $\Gterminal$, consider the following procedure to construct $\matchingstar(\economyrt)$ for each truncation $\economyrt$ of $\economyrterminal$. For each $1\leq t\leq\terminal-1$, having defined $(\matchingstar(\economyr^1),\dots,\matchingstar(\economyr^{t-1}))$, let $\cala(\matchingstaruptot,\economyrt)\cup\calb(\matchingstaruptot,\economyrt)$ denote the agents who can match at \economyrt\  (if \periodone,\ then $\cala(\emptyset,\economyr^1)=A_1,\calb(\emptyset,\economyr^1)=B_1$.) For each $k\in\cala(\matchingstaruptot,\economyrt)\cup\calb(\matchingstaruptot,\economyrt)$, let $\matchingb_{\economyrt}^k$ denote the matching $\matching\in\Matchingsfinal(\matchingstaruptot,\economyrt)$ that minimizes $k$'s payoff\footnote{That is, if $k=a$, I am minimizing $U(a,\cdot, \economyrt)$, while if $k=b$, I am minimizing $V(b,\cdot,\economyrt)$.} at $\economyrt$ and satisfies that
\begin{enumerate}
\item $\matching(\economyrt)(k)=k$
\item $(\matching(\economyr^{t+\dateindex}))_{s\in\{1,\dots,\terminal-t\}}\in\ds_{T-t}(G_{\terminal-t}(\matching^t,\economyrt))$.
\item $\matching(\economyrt)$ is stable amongst those who match in period $t$.
\end{enumerate}
By the inductive hypothesis, $\ds_\dateindex$ is non-empty valued for $\dateindex<\terminal$. Since $t\geq 1$, then $\ds_{\terminal-t}(\cdot)$ is non-empty. Thus, for all $k\in\cala(\matchingstaruptot,\economyrt)\cup\calb(\matchingstaruptot,\economyrt)$, $\matchingb_{\economyrt}^k$ is well-defined since the above set of matchings is non-empty and finite.\footnote{Recall that any period-t matching that is (statically) stable for the economy formed by $\cala(\matchingstaruptot,\economyrt)\cup\calb(\matchingstaruptot,\economyrt)\setminus\{k\}$ can be extended to a period-t matching for the agents in $\cala(\matchingstaruptot,\economyrt)\cup\calb(\matchingstaruptot,\economyrt)$ so that agent $k$ is unmatched and the period-t matching is stable amongst those who match at \economyrt. }
\newline\indent For each $\agenta\in\cala(\matchingstaruptot,\economyrt)$, consider the preference list on $\calb(\matchingstaruptot,\economyrt)$ defined as follows: $\agentb\in \calb(\matchingstaruptot,\economyrt)$ is acceptable to $\agenta$ only if $\genericu\geq U(\agenta,\matchingb_{\economyrt}^\agenta,\economyrt)$ and is not acceptable to $\agenta$ otherwise. Moreover, if $\agentb,\agentbp\in\calb(\matchingstaruptot,\economyrt)$ are acceptable to $\agenta$, then they are ordered according to $\genericuo\cdot)$ in $a$'s preference list.

 Similarly, for $\agentb\in\calb(\matchingstaruptot,\economyrt)$, consider the preference list on $\cala(\matchingstaruptot,\economyrt)$ defined as follows: $\agenta\in\cala(\matchingstaruptot,\economyrt)$ is acceptable to $\agentb$ only if $\genericv\geq V(\agentb,\matchingb_{\economyrt}^\agentb,\economyrt)$ and is not acceptable to $\agentb$ otherwise. Moreover, if $\agenta,\agentap\in\cala(\matchingstaruptot,\economyrt)$ are acceptable to $\agentb$, then they are ordered according to $v(\cdot,\agentb)$ in $\agentb$'s preference list.
\newline\indent Define $\matchingstar(\economyrt)$ to be the outcome of deferred acceptance on $\langle\cala(\matchingstaruptot,\economyrt),\calb(\matchingstaruptot,\economyrt)\rangle$ using the above preference lists.

 For $t=\terminal$, let $\matchingstar(\economyrterminal)$ coincide with $\matchingstar(\economyr^{\terminal-1})$ for all those who have matched at $\economyr^{\terminal-1}$ and 
with the outcome of running deferred acceptance with agents on $\cala(\matchingstaruptoterminal,\economyrterminal)$ proposing to agents on $\calb(\matchingstaruptoterminal,\economyrterminal)$, elsewhere.
\newline\indent Note that this procedure on all $\economyrt$ and all $1\leq t\leq\terminal$ implies that \matchingstar\ is an element of $\Matchings_\terminal$. I now make some observations about $\matchingstar$.

First, for all periods $t\in\{1,\dots,\terminal\}$, and all realizations \economyrt,\ $\matchingstar(\economyrt)$ is stable amongst those who match at \economyrt.\ Second, for all periods $t\in\{1,\dots,\terminal-1\}$, and all realizations \economyrt,\ there is no $k\in\cala(\matchingstaruptot,\economyrt)\cup\calb(\matchingstaruptot,\economyrt)$ such that $\matchingstar(\economyrt)(k)\neq k$ and for all $\matchingb\in M_{\ds}(k,\matchingstar,\economyrt)$, $k$ prefers $\matchingb$ to $\matchingstar$. To see that this holds note the following. By construction $\matchingb_{\economyrt}^k$ can always be taken to be an element of $M_{\ds}(k,\matchingstar,\economyrt)$ by having it coincide with $\matchingstar$ at all realizations $\economyr^\dateindex$ that do not follow \economyrt.\ Moreover, if $k$ is matched at $\economyrt$, it must be that their matching partner is (weakly) preferred to $\matchingb_{\economyrt}^k$.
\newline\indent Fix $t=\terminal$ and note that for all $\economyrterminal$, $\matchingstar(\economyrterminal)$ is pairwise stable and individually rational. Let $t\leq\terminal-1$ and let $\economyrt$ denote a realization. Suppose that one has shown that $(\matchingstar(\economyr^{t+\dateindex}))_{\dateindex\in\{1,\dots,\terminal-t\}}$ is pairwise stable starting from period $t+1$ onward. Note that, in particular, this implies that $(\matchingstar(\economyr^{t+\dateindex}))_{\dateindex\in\{1,\dots,\terminal-t\}}\in\ds_{T-t}(G_{\terminal-t}(\matching^{\star^t},\economyrt))$. I now show that there are no pairwise blocks involving agents that have arrived by \economyrt.\ Toward a contradiction, suppose there exists a pair $(\agenta,\agentb)\in\cala(\matchingstaruptot,\economyrt)\times\calb(\matchingstaruptot,\economyrt)$ such that
\begin{align*}
\genericu>U(a,\matchingstar,\economyrt), \genericv>V(\agentb,\matchingstar,\economyrt).
\end{align*}
Since $\matchingstar(\economyrt)$ is stable amongst those who match in period $t$, then it must be that either $\agenta$ or $\agentb$ are unmatched at $\matchingstar(\economyrt)$. Without loss of generality, suppose that it is \agentb.\ Thus, \agentb\ declared \agenta\ as unacceptable in the truncated preference lists  so that
\begin{align}
\genericv<V(\agentb,\matchingb_{\economyrt}^\agentb,\economyrt).
\end{align}
Combining the above inequalities, it follows that
\begin{align}\label{eq:wronglowerboundt}
V(\agentb,\matchingstar,\economyrt)<V(\agentb,\matchingb_{\economyrt}^\agentb,\economyrt).
\end{align}
Note, however, that $\matchingstar\in M_{\ds}(\agentb,\matchingstar,\economyrt)$.\footnote{This, of course, uses that we have shown that $\matchingstar(\economyrt,\cdot)$ is pairwise stable and  has no blocks by waiting.} This, together with \autoref{eq:wronglowerboundt}, contradicts the definition of $\matchingb_{\economyrt}^\agentb$. Thus, $\matchingstar$ is pairwise stable from period $t$ onward. Proceeding this way, it follows that $\matchingstar$ is pairwise stable. 
This proves $\mathbf{P(\terminal)}=1$.
\subsection{Failure of the Lone Wolf theorem}\label{appendix:lattice}
The result in \autoref{prop:lattice} is based on the following example:
\begin{ex}\label{ex:lone-wolf}
Let $\terminal=2$. Arrivals are as follows: $A_1=\{a_{11},a_{12}\}$, $A_2=\{a_{21},a_{22}\}$, $B_1=\{b_{11}\}$, $B_2=\{b_{21},b_{22}\}$. Preferences are given by

\begin{table}[h!]
\begin{tabular}{lcccc}
$a_{11}$:&$(b_{21},1)$&$(b_{22},0)$&$(b_{11},0)$&$(b_{22},1)$\\
$a_{12}$:&$(b_{11},0)$&&&\\
$a_{21}$:&$(b_{21},0)$&&&\\
$a_{22}$:&$(b_{22},0)$&$(b_{21},0)$&&
\end{tabular}\hspace{1cm}
\begin{tabular}{lcccc}
$b_{11}$:&$(a_{11},0)$&$(a_{12},0)$&&\\
$b_{21}$:&$(a_{22},0)$&$(a_{11},0)$&$(a_{21},0)$&\\
$b_{22}$:&$(a_{11},0)$&$(a_{22},0)$&&\\
\textcolor{white}{$b_{22}$:}&\textcolor{white}{$(a_{11},0)$}&\textcolor{white}{$(a_{22},0)$}&&
\end{tabular}
\end{table}
The following two matchings are dynamically stable:
\begin{align*}
\matching^L=\left(\begin{array}{lcl}a_{11}&\_&b_{11}\\\hline a_{21}&\_&b_{21}\\a_{22}&\_&b_{22}\\a_{12}&\_&\emptyset\end{array}\right)\hspace{1cm}\matching^R=\left(\begin{array}{lcl}a_{12}&\_&b_{11}\\\hline a_{11}&\_&b_{21}\\a_{22}&\_&b_{22}\\a_{21}&\_&\emptyset\end{array}\right),
\end{align*}
Agent $a_{12}$ is unmatched under $\matching^L$, while agent $a_{21}$ is unmatched under $\matching^R$, so that both the Lone Wolf Theorem and the lattice property fail in this economy. Note that both matchings satisfy condition \ref{itm:ds1}. In particular, $\agentb_{11}$ prefers $\agenta_{11}$ to $\agenta_{12}$ so that $\agenta_{12}$ has no possibility of being matched under $\matching^L$, while $\agentb_{21}$ prefers $\agenta_{11}$ to $\agenta_{21}$, so that $\agenta_{21}$ has no possibility of matching under $\matching^R$. Below, $\matchingb^L$ (resp., $\matchingb^R$) denotes the conjecture that dissuades $\agenta_{11}$ (resp., $\agentb_{11}$) from blocking $\matching^L$ (resp., $\matching^R$) in \periodone:
\begin{align*}
\matchingb^L=\left(\begin{array}{lcl}a_{12}&\_&b_{11}\\\hline a_{11}&\_&b_{22}\\a_{22}&\_&b_{21}\\a_{21}&\_&\emptyset\end{array}\right)\quad\quad\matchingb^R=\left(\begin{array}{lcl}&\emptyset&\\\hline a_{11}&\_&b_{21}\\a_{22}&\_&b_{22}\\a_{12}&\_&b_{11}\\a_{21}&\_&\emptyset\end{array}\right),
\end{align*}
so that neither $\agenta_{11}$ may block $\matching^L$ by remaining unmatched, nor can $\agentb_{11}$ block $\matching^R$ by remaining unmatched.
\end{ex}
\section{Proofs of Section \ref{sec:timing}}\label{appendix:sequential}
\begin{proof}[Proof of \autoref{prop:stoch-queueing}]
Let \Gterminal\ and \matching\ satisfy the assumptions of \autoref{prop:stoch-queueing}. Let $t\leq\terminal-1$ denote the largest $\dateindex\leq\terminal-1$ such that there exists a realization \economyrs\ such that either Conditions \ref{itm:ds2} or \ref{itm:ds3} of \autoref{definition:ds} fail. Letting $\economyrt=(\economyr^{t-1},A_t,B_t)$ denote the realization for which either Conditions \ref{itm:ds2} or \ref{itm:ds3} of \autoref{definition:ds} fail, it follows that there exists $k\in\cala(\matchinguptot,\economyrt)\cup\calb(\matchinguptot,\economyrt)$ such that $\matching(\economyrt)(k)\neq k$ that can improve by remaining unmatched at \economyrt.\ Moreover, since \matching\ is not dynamically stable for arriving agents it follows that $\agent\in A_t\cup B_t$. Without loss of generality, assume that $k=\agenta\in A_t$. In a slight abuse of notation, let $\economyrt\setminus\{\agenta\}$ denote $(\economyr^{t-1},A_t\setminus\{a\},B_t)$. Note that $\cala(\matchinguptot, \economyrt)\setminus\{\agenta\}=\cala(\matchinguptot,\economyrt\setminus\{\agenta\})$ and $\calb(\matchinguptot,\economyrt)=\calb(\matchinguptot,\economyrt\setminus\{\agenta\})$

 I now define a matching, \matchingb,\ and argue that $\matchingb\in M_\ds(\agenta,\matching^{t-1},\economyrt)$. First, let $\matchingb\in\Matchingsfinal(\matching^{t-1},\economyrt)$. Second, let
\[\matchingb(\economyrt)(k^\prime)=\left\{\begin{array}{ll}\matchingb(\economyrt\setminus\{\agenta\})(k^\prime)&\text{if }k^\prime\neq a\\
a&\text{otherwise}
\end{array}\right.\]
Finally, for each $1\leq\dateindex\leq\terminal-t$ and realizations $\economyr^{t+\dateindex}=(\economyrt,\economyr_{t+1},\dots,\economyr_{t+\dateindex})$, let
\[\matchingb(\economyrs)(k)=\matchingb(\economyrt\setminus\{a\},\economyr_{t+1}\cup\{\agenta\},\dots,\economyr_{t+\dateindex})(k)\]
for each $k\in\cala(\matchingb^{t+\dateindex-1},\economyr^{t+\dateindex})\cup\calb(\matchingb^{t+\dateindex-1},\economyr^{t+\dateindex})$. 
Note that $\matchingb(\economyrt)$ is stable amongst those who match at \economyrt\ by assumption. Moreover, $\agenta$ is unmatched at \economyrt.\ It remains to be argued that $(\matchingb(\economyr^{t+\dateindex}))_{\dateindex\in\{1,\dots,\terminal-t\}}\in\ds_{\terminal-t}(G_{\terminal-t}(\matchingb,\economyrt))$. Note that by definition $(\matching(\economyrt\setminus\{a\},\economyr_{t+1}\cup\{\agenta\},\dots,\economyr_{t+\dateindex}))_{s\in\{1,\dots,\terminal-t\}}$ satisfies the definition of dynamic stability under distribution
$G_{\terminal}(\cdot|\economyr^t\setminus\{a\},\economyr_{t+1}\cup\{\agenta\},\dots,\economyr_{t+\dateindex})$. Exchangeability implies that this is the same as $G_{\terminal}(\cdot|\economyr^t,\economyr_{t+1},\dots,\economyr_{t+\dateindex}))$. Thus, $\matchingb\in M_\ds(\agenta,\matching,\economyrt)$.

Since \matching\ is not dynamically stable at \economyrt,\ it follows that for all matchings in $M_\ds(\agenta,\matching,\economyrt)$, $U(\agenta,\cdot,\economyrt)>U(\agenta,\matching,\economyrt)$. In particular,
\begin{align*}
U(\agenta,\matchingb,\economyrt)>U(\agenta,\matching,\economyrt).
\end{align*}
The result then follows.
\end{proof} 

\begin{proof}[Proof of \autoref{prop:ds-sp}]
Let \Gtwo\ and \mada\ be as in the statement of \autoref{prop:ds-sp}. Let $\agenta\in\arrivalaupto_1(\economyr^2)$ be such that $\mada(\economyr^1)(a)\neq a$. Construct a matching \matchingb\ as follows. Run deferred acceptance as if the economy was given by$(A_1\setminus\{\agenta\},B_1,A_2\cup\{a\},B_2)$, that is:
\begin{enumerate}[label=(\alph*)]
\item $\agenta$ makes proposals as if \agenta\ arrived in $\periodtwo$ That is, \agenta\ proposes to $\agentb\in\arrivalbupto_2(\economyr^2)$ before $\agentbp\in\arrivalbupto_2(\economyr^2)$ if, and only if, $\genericu>\genericuo,\agentbp)$
\item $b\in\arrivalbupto_1(\economyr^2)$ accepts $a$'s offer over $\agentap \in A_1$ only if $\delta_\agentb\genericv>v(\agentap,\agentb)$.
\end{enumerate}
Note that $\matchingb\in M_\ds(\agenta,\mada,\economyr^1)$. Toward a contradiction, suppose that $U(\agenta,\matchingb,\economyr^1)>U(\agenta,\mada,\economyr^1)$. Let $A^+$ denote the set of agents on side $A$ who prefer \matchingb\ to \mada.\ According to a lemma by J.S. Hwang (see \cite{gale1985some} for a proof), there exists $\agentap\notin A^+$ and $b\in\matchingb(\economyr^2)(A^+)$ that block \matchingb\ such that $\delta_{\agentap}^{\mathbbm{1}[\agentb\in\arrivalb_2,\agentap\in\arrivala_1]}u(\agentap,\agentb)>U_{\cdot}(\agentap,\matchingb)$ and $\delta_\agentb^{\mathbbm{1}[\agentb\in\arrivalb_1,\agentap\in\arrivala_2]}v(\agentap,\agentb)>V_\cdot(\agentb,\matchingb)$. This contradicts the definition of \matchingb. Thus, it cannot be that \agenta\ strictly prefers \matchingb\ to \mada.\
\end{proof}
\subsection{Proof of \autoref{theorem:ds-implementation}}\label{appendix:sequential-implementation}
I introduce notation to define the agents' strategies and the solution concept. A public history at the beginning of period $t$, \publict, is a realization \economyrt, and a matching through period $t-1$ for that realization $\matchingduptot=(\matchingd(\economyrone),\dots,\matchingd(\economyruptot))$. 

Let $\dateindex\geq 1$, fix a realization $\economyr^\dateindex=(\arrivala_1,\arrivalb_1,\dots,\arrivala_\dateindex,\arrivalb_\dateindex)$, and let $\agent\in\arrivala_\dateindex\cup\arrivalb_\dateindex$. For any period $t\geq\dateindex$, agent \agent\ would have observed a public history $(\economyrt, \matchingduptot)$, together with their decision to participate and their reported preferences. While agent \agent\ can condition their period $t$ strategy both on \publict\ and their past actions, it will be clear from the proof that it is without loss of generality to focus on strategies that condition on the public history alone. Thus, to keep notation simple, I define agent \agent's behavioral strategy in period $t$ as a mapping $\strat_{\agent,t}$, that takes a public history \publict\ such that \agent\ is unmatched in period $t$ and outputs \agent's decision to participate and a ROL, $\rol_{\agent,t}$. Let $\strat_\agent=(\strat_{\agent,t})_{t\geq\dateindex}$ denote \agent's strategy profile.

A pure strategy profile $\strat=(\strat_\agent)_{\agent\in\arrivalauptoterminal\cup\arrivalbuptoterminal}$ together with a realization, \economyrterminal, determine a terminal history $h_\strat^{\terminal+1}=(\economyrterminal,\matchingd^\strat)$, where letting $\arrivalab_t^\strat\cup\arrivalbb_t^\strat$ denote the set of agents who participate in the spot mechanism in period $t$ under \strat\ when the realization is \economyrterminal, we have that
\begin{align*}
\matchingd(\economyrt)^\strat|_{\arrivalab_t^\strat\cup\arrivalbb_t^\strat}=\spot_t(\arrivalab_t^\strat,\arrivalbb_t^\strat,(\strat_\agent(\publict_\strat))_{\agent\in\arrivalab_t^\strat\cup\arrivalbb_t^\strat})
\end{align*}
while $\matchingd^\strat(\economyrt)(\agent)=\agent$ for all $\agent\in\cala(\matchingd^{\strat,t-1},\economyrt)\cup\calb(\matchingd^{\strat,t-1},\economyrt)\setminus(\arrivalab_t^\strat\cup\arrivalbb_t^\strat)$. Similarly, a pure strategy profile \strat\ together with a public history $\publict=(\economyrt,\matchingduptot)$ determine a matching $\matching^\strat(\economyrt,\matchingduptot)\in\Matchingsfinal(\matchingd,\economyrt)$. An agent's payoff from strategy profile \strat\ at public history $\publict=(\economyrt, \matchingduptot)$ is determined by the payoff from the matching it induces. That is, letting $\agenta\in\cala(\matchingduptot,\economyrt)$
\begin{align*}
U(\agenta,\strat|\publict)=U(\agenta,\matchingd^{\strat}(\matchingduptot),\economyrt),
\end{align*}
and similarly for $\agentb\in\calb(\matchingduptot,\economyrt)$
\begin{definition}\label{definition:pairwise-spne}
A pure strategy profile $\strat=(\strat_\agent)_{\agent\in\arrivalauptoterminal\cup\arrivalbuptoterminal}$ is a pairwise SPNE of \gamestable\ if it is a SPNE of \gamestable\ and the following holds. There is no period $t\geq 1$, public history $\publict=(\economyrt,\matchingduptot)$, and pair $(\agenta,\agentb)$ such that:
\begin{enumerate}
    \item $\pair\in\cala(\matchingduptot,\economyrt)\times\calb(\matchingduptot,\economyrt)$,
    \item There exists $\stratb_\agenta,\stratb_\agentb$ such that $U(\agenta,(\strat_{-(a,b)},\stratb_\agenta,\stratb_\agentb)|\publict)>U(\agenta,\strat|\publict)$ and $\newline V(\agentb,(\strat_{-(a,b)},\stratb_\agenta,\stratb_\agentb)|\publict)>V(\agentb,\strat|\publict)$
\end{enumerate}
\end{definition}

\begin{lemma}\label{lemma:join}
In both games, it is without loss of generality to focus on strategy profiles where the agents participate in the mechanism whenever they are unmatched.
\end{lemma}
 This follows from noting that (i) joining is not observable, and (ii) agents \agent\ who participate can submit a ROL that only includes themselves, i.e., an empty ROL.
\begin{lemma}\label{lemma:no-reshuffling}
In both games, without loss of generality, agents either submit an empty ROL or a truncation of the ranking induced by their Bernoulli utility function.
\end{lemma}
\begin{proof}
Fix a history $\publict=(\economyrt,\matchingduptot)$ and an agent $\agent\in\cala(\matchingduptot,\economyrt)\cup\calb(\matchingduptot,\economyrt)$. There are two cases to consider. First, assume that given the equilibrium strategy at \publict, \agent\ is matched at the end of period $t$. Then the result in \cite{roth1991incentives} implies that \agent\ can do weakly better by submitting a truncation. Second, suppose instead that given the equilibrium strategy at \publict, \agent\ remains unmatched at the end of period $t$.  Then it is a property of stable matching algorithms that the set of agents other than \agent\ who are unmatched is independent of the ROL submitted by \agent, as long as \agent\ is unmatched. 
%
%
To see this, let $\rol$ denote the submitted ROLs under \strat\ at \publict. Furthermore, let \rolb\ coincide with \rol\ for $\agentp\in\cala(\matchingduptot,\economyrt)\cup\calb(\matchingduptot,\economyrt)\setminus\{\agent\}$. Let $\matching(\economyrt)=\spot_t(\cala(\matchingduptot,\economyrt),\calb(\matchingduptot,\economyrt),\rol)$ and $\matchingb(\economyrt)=\spot_t(\cala(\matchingduptot,\economyrt),\calb(\matchingduptot,\economyrt),\rolb)$. 
%
%
Assume that $\matching(\economyrt)(\agent)=\matchingb(\economyrt)(\agent)=\agent$. Toward a contradiction, suppose there exists $\agenta\in\cala(\matchingduptot,\economyrt)$ such that $\matching(\economyrt)(\agenta)\neq\agenta=\matchingb(\economyrt)(\agenta)$. Let $A^-=\{\agentap\in\cala(\matchingduptot,\economyrt):u(\agentap,\matching(\economyrt)(\agentap))>u(\agentap,\matchingb(\economyrt)(\agentap))\}$ and $B^+=\{\agentb\in\calb(\matchingduptot,\economyrt):v(\matchingb(\economyrt)(\agentb),\agentb)>v(\matching(\economyrt)(\agentb),\agentb)\}$. The Decomposition Lemma (\citealp{roth1992two}) implies that $\matching(\economyrt)(\arrivala^-)=\matchingb(\economyrt)(\arrivala^-)=\arrivalb^+$, a contradiction. Thus, without loss of generality, \agent\ can submit an empty list in that case. Since the ROL that \agent\ submits at $h^t$ is unobserved, one can then change \agent's strategy at $h^t$ without affecting the equilibrium.
\end{proof}
\begin{lemma}
In both games, for any period $t$ and public history $\publict=(\economyrt,\matchingduptot)$, the continuation matching $(\matching^\strat(\publict)(\economyr^{t+\dateindex}))_{\dateindex=0}^{T-t}$ is pairwise stable and individually rational for $G_{\terminal-(t-1)}(\matchingduptot,\economyrt)$
\end{lemma}
\begin{proof}
Individual rationality follows from the ability of the agents to remain single through period \terminal\ by always submitting empty ROLs. 

In \gamestable, pairwise stability follows from the possibility of joint deviations: otherwise, there would be a period $\dateindex\geq t$, a realization $\economyrs$ that follows \economyrt, and a pair $\pair\in\cala(\matching^{\strat,\dateindex-1},\economyr^{\dateindex})\cup\calb(\cdot)$ that can deviate by submitting lists that only list each other. Any stable mechanism matches \pair\ together. 

In \gameda, pairwise stability follows from the properties of deferred acceptance. Suppose there is a period $\dateindex\geq t$, a realization $\economyr^{\dateindex}$, and a pair $\pair\in\cala(\matching^{\strat,\dateindex-1},\economyr^{\dateindex})\cup\calb(\cdot)$ that prefer to match with each other over their outcome under $\matching^\strat(\publict)$. If $\matching^\strat(\publict)(\economyr^\dateindex)(\agenta)\neq\agenta$, then this contradicts either that \agenta\ submitted a truncation of their true ranking or the stability of the deferred acceptance algorithm with respect to the reported preferences. Suppose then that $\matching^\strat(\publict)(\economyr^\dateindex)(\agenta)=\agenta$. The stability of the deferred acceptance with respect to the reported preferences implies that \agenta\ did not include \agentb\ in their ROL. Consider the following strategy for \agenta: \agenta\ submits $\tilde{\rol}_\agenta$ listing all $\agentbp\in\sideb$ that are preferred to \agenta's outcome under $\matching^\strat(\publict)$. Under this ROL, which includes \agentb, it cannot be that \agenta\ is unmatched in period \dateindex. Otherwise, \agentb\ would still be matched to $\matching^\strat(\publict)(\economyr^\dateindex)(\agentb)$, a contradiction (recall the proof of \autoref{lemma:no-reshuffling}).  Then, $\matching^\strat(\publict)(\economyr^\dateindex)(\agenta)\neq\agenta$, so that \agenta\ has a profitable deviation, contradicting the definition of SPNE.
\end{proof}
\begin{cor}\label{cor:terminal-stability}
For any public history $h^{\terminal}=(\matchingd^{\terminal-1},\economyrterminal)$, the matching $\matching^\strat(h^{\terminal})(\economyrterminal)$ satisfies \autoref{definition:static-stability} for $\cala(h^{\terminal})\cup\calb(h^{\terminal})$.
\end{cor}
%
%
%
\begin{lemma}
In \gameda, for any period $t$ and public history $\publict=(\economyrt,\matchingduptot)$, the matching $\matching^\strat(\publict)$ cannot be improved upon by agents on side \arrivala\ waiting to be matched. Similarly, in \gamestable, for any period $t$ and public history $\publict=(\economyrt,\matchingduptot)$, the matching $(\matching^\strat(\publict)(\economyr^{t+\dateindex}))_{\dateindex=0}^{\terminal-t}\in\ds_{\terminal-(t-1)}(G_{\terminal-t}(\publict))$.
\end{lemma}
\begin{proof}
The proof is similar for both games so I focus on \gamestable. Let $t\leq\terminal-1$ denote the largest $\dateindex\leq\terminal-1$ such that there exists $h^\dateindex$ such that $\matching^\strat(h^\dateindex)$ fails conditions \ref{itm:ds2}-\ref{itm:ds3} in \autoref{definition:ds}.\footnote{\autoref{cor:terminal-stability} implies that $t\leq\terminal-1$ is well-defined.} Let $\publict=(\economyrt,\matchingduptot)$ denote one such history and let  $\agent\in\cala(\publict)\cup\calb(\publict)$ be such that \agent\ prefers all matchings in $M_\ds(\agent,\matching^\strat(\publict),\economyrt)$ over $\matching^\strat(\publict)$. 

Consider the matching $\matchingb(\publict)\in\Matchingsfinal(\matching^\strat,\economyrt)$ that arises when \agent\ deviates and submits an empty ROL and then continues to play the equilibrium strategy. 

By definition of $t,\publict$, $(\matchingb(\economyr^{t+\dateindex}))_{\dateindex=1}^{\terminal-t}\in\ds_{\terminal-t}(G_{\terminal-t}(\matchingbt,\economyrt))$. It remains to show that $\matchingb(\economyrt)$ satisfies \autoref{definition:ds-static-stability}. Toward a contradiction, suppose there exists a pair \pair\ such that $\matchingb(\economyrt)(\agenta)\neq\agenta$ and $\matchingb(\economyrt)(\agentb)\neq\agentb$, but $\genericu>u(\agenta,\matchingb(\economyrt)(\agenta))$ and $\genericv>v(\matchingb(\economyrt)(\agentb),\agentb)$. Then, it must be that the ROLs $(\succ_\agenta,\succ_\agentb)$ satisfy that either $\matchingb(\economyrt)(\agenta)\succ_\agenta\agentb$ or $\matchingb(\economyrt)(\agentb)\succ_b\agenta$. This contradicts \autoref{lemma:no-reshuffling}: under truncations, the agents do not switch the order of agents on the other side relative to their true preferences. It follows that $\matchingb\in M_\ds(\agent,\publict)$. Thus, \agent\ has a deviation at public history \publict, contradicting the definition of SPNE. 
\end{proof}
\section{General model}\label{sec:general}
I now introduce notation that allows me to keep track of the date each agent is born so that the possibility of there being individuals who share the same characteristics within a period or across periods does not create confusion.

A realization of length 1 is a tuple $E=\economylmps$, where $I,J\subseteq\mathbb{N}$ are (possibly empty) finite sets of indices and $\typea:I\mapsto\sidea$ and $\typeb:J\mapsto\sideb$ denote the characteristics of agents in $I$ and $J$, respectively. Let $\economy$ denote the set of all such tuples. Formally, 
\begin{align*}
\economy=\{\economylmps:I,J\subseteq\mathbb{N},|I|,|J|<K,\typea\in\sidea^I,\typeb\in\sideb^J\}
\end{align*} 
The bound $K\in\mathbb{N}$ in the definition of \economy\ ensures that all sets of matchings that I refer to in what follows are finite.

For any $t\geq 2$, let $\overline{\economy}^t$ denote the subset of $\economyt$ that satisfies the following property. A tuple $(\economyslmps)_{\dateindex=1}^t\in\overline{\economy}^t$ only if $I_\dateindex\cap I_{\dateindex^\prime}=\emptyset$ and $J_\dateindex\cap J_{\dateindex^\prime}=\emptyset$ whenever $\dateindex\neq\dateindex^\prime$.\footnote{Note that we can always take the indices to be disjoint without loss of generality by consecutively enumerating the agents that arrive on each side.} An economy of length $\terminal$ is a distribution $\Gterminal$ on $\supportterminal$. In what follows, if $\dateindex\leq t$, and $\economyrt=(\economyrs,\economyr^{t-s})\in\supportt$, I say that $\economyrs$ is a truncation of \economyrt,\ and that \economyrt\ follows \economyrs.\

Fix a tuple $\economyrt\in\supportt$. For $\dateindex\leq t$, let $(I_s(\economyrt),\typeas(\economyrt),J_s(\economyrt),\typebs(\economyrt))$ denote the $s^{th}$ element of \economyrt.\  Let $\arrivaliuptos(\economyrt)$ denote the implied arrivals on side $A$ through period $\dateindex\leq t$; similarly, let $\arrivaljuptos(\economyrt)$ denote the implied arrivals on side $B$ through period $\dateindex\leq t$. Let $\typeauptot(\economyrt)$ be the mapping from $\arrivaliuptot(\economyrt)$ to $\sidea$ that coincides with $\typea_{\dateindex}(\economyrt)$ on $I_\dateindex(\economyrt)$ for all $1\leq\dateindex\leq t$. Define $\typebuptot(\economyrt)$ in a similar way. To economize on notation, I do not make the dependence on \economyrt\ explicit if it is unlikely to create confusion.

A matching $\matching$ associates to each element $\economyrt\in\supportt$ a period $t$ matching
and satisfies the following property. For any $t\in\{1,\dots,\terminal\}$ and any $\economyrt\in\supportt$, for any $i\in\arrivaliuptot(\economyrt)$, if $\matching(\economyrt)(i)\neq i$, then for any $t\leq\dateindex$ and $\economyrs$ that follows \economyrt,\ $\matching(\economyrs)(i)=\matching(\economyrt)(i)$. 
Denote by $\Matchings_\terminal$ the set of all matchings for an economy of length \terminal.\ 

Fix a matching \matching\ and a realization of the arrivals through period $t$, $\economyrt=(\economyr^{t-1},\economytlmps)$. This determines the agents who are available to match at $(t,\economyrt)$
\begin{align*}
\cali(\matchinguptot,\economyrt)&=\{i\in\arrivaliupto_{t-1}(\economyrt):\matching_{t-1}(\economyr^{t-1})(i)=i\}\cup I_t\\
\calj(\matchinguptot,\economyrt)&=\{j\in\arrivaljupto_{t-1}(\economyrt):\matching_{t-1}(\economyr^{t-1})(j)=j\}\cup J_t.
\end{align*}
A matching \matching\ and a realization $\economyr^{t-1}$ also define a continuation economy of length $\terminal-t+1$, where the distribution of arrivals $G_{\terminal-t+1}(\matchinguptot,\economyr^{t-1})$ assigns probability $\Gterminal(\economyrt,(\economyslmps)_{\dateindex=t+1}^\terminal|\economyr^{t-1})$ to the length $T-t+1$ economy
\[((\cali(\matchinguptot,\economyrt),\typeauptot(\economyrt),\calj(\matchinguptot,\economyrt),\typebuptot(\economyrt)),(\economyslmps)_{\dateindex=t+1}^\terminal),\]
and $0$ otherwise.\footnote{Technically, $\typeauptot$ maps \arrivaliuptot\ to \sidea.\ Thus, one should interpret \typeauptot\ in the definition of $\economyr^{\geq t+1}$ as the restriction of \typeauptot\ to $\cali(\matchinguptot)$.}

To close the model, I define agents' preferences over matchings. Fix a matching, $\matching$, and a realization \economyrt.\ Let $i\in\cali(\matchinguptot,\economyrt)$. For each $\economyrterminal=(\economyrt,\cdot)$, let $t_\matching(i,\economyrterminal)$ denote the smallest $t\leq\dateindex$ such that $\matching(\economyrs)(i)\neq a$ and $\economyrterminal$ follows $\economyrs$; otherwise, let $t_\matching(i,\economyrterminal)=\terminal$. Then, let
\[U(i,\matching,\economyrt)=\mathbb{E}_{G(\cdot|\economyrt)}[\delta_{\typeauptot(i)}^{t_\matching(i,\economyrt,\cdot)-t}u(\typeauptot(i),\typebuptoterminal(\matching(\economyrt,\cdot)(i)))],\]
denote $i$'s payoff from matching \matching\ at date $t$ when arrivals are \economyrt.\ In the above expression, to economize on notation (i) I omit the dependence of both \typeauptot\ and \typebuptoterminal\ on $\economyrt$ and \economyrterminal,\ respectively, and (ii) I abuse notation so that if $\matching(\economyrterminal)(i)=i$, then $\typebuptoterminal(i)=\typeauptot(i)$. Similarly, for $j\in\calj(\matchinguptot,\economyrt)$, 
let
\[V(j,\matching,\economyrt)=\mathbb{E}_{G(\cdot|\economyrt)}[\delta_{\typebuptot(j)}^{t_\matching(j,\economyrt,\cdot)-t}v(\typeauptoterminal(\matching(\economyrt,\cdot)(j)),\typebuptot(j))].\]
\subsection{Existence}
For completeness, I replicate in this section the proof of Theorem 1 in the main text using the notation in the previous section. 

The proof proceeds by induction on $T\geq 2$. As argued in the main text $\ds_1$ coincides with the set of pairwise stable matchings for the static economy. Let $\mathbf{P(\terminal)}$ denote the following inductive statement:

$\mathbf{P(T):}$ The correspondence $\ds_\terminal$ is non-empty valued. 

I first prove $\mathbf{P(2)}=1$. Let $\Gtwo$ denote an economy of length $\terminal=2$. For each $\economyr^1=(I_1,\typea_1,J_1,\typeb_1)$ on the support of $\Gtwo$, I construct a matching $\matchingstar(\economyr^1)$ as follows.
For each $k\in I_1\cup J_1$, let $\matchingb_{\economyr^1}^k$ denote the payoff minimizing matching \matching\ such that
\begin{enumerate}
\item $\matching(\economyr^1)(k)=k$,
\item\label{itm:p22} For each $\economyr^2=(\economyr^1,I_2,\typea_2,J_2,\typeb_2)$ in the support of $\Gtwo$, $\matching(\economyr^2)$ is stable for $\langle\cali(\matching^1,\economyr^2),\calj(\matching^1,\economyr^2)\rangle$,
\item\label{itm:p23} $\matching(\economyr^1)$ is stable amongst those who match at $\economyr^1$.
\end{enumerate}
Note that the set of matchings \matching\ that satisfy these three requirements is non-empty since the set of stable matchings in period 2 is non-empty and a matching that leaves everyone unmatched in period 1 is stable amongst those who match $\economyr^1$. Also note that \autoref{itm:p22} implies that any such matching \matching\ satisfies that $\matching(\economyr^1,\cdot)\in\ds_1(\Gtwo(\matching^1,\economyr^1))$.

For each $\agenti\in\arrivali_1$, consider the preference list on $\arrivalj_1$ defined as follows: $\agentj\in \arrivalj_1$ is acceptable to $\agenti$ only if $\genericuione\geq U(\agenti,\matchingb_{\economyr^1}^\agenti,\economyr^1)$ and is not acceptable to $\agenti$ otherwise. Moreover, if $\agentj,\agentjp\in\arrivalj_1$ are acceptable to $\agenti$, then order them according to $u(\typea_1(i),\cdot)$ in $\agenti$'s preference list.

Similarly, for $\agentj\in\arrivalj_1$, consider the preference list on $\arrivali_1$ defined as follows: $\agenti\in\arrivali_1$ is acceptable to $\agentj$ only if $\genericvjone\geq V(\agentj,\matchingb_{\economyr^1}^\agentj,\economyr^1)$ and is not acceptable to $\agentj$ otherwise. Moreover, if $\agenti,\agentip\in\arrivali_1$ are acceptable to $\agentj$, then order them according to $v(\cdot,\typeb_1(\agentj))$ in $\agentj$'s preference list.

Define $\matchingstar(\economyr^1)$ so that it coincides with the outcome of running the deferred acceptance algorithm with agents in $\arrivali_1$ proposing to agents in $\arrivalj_1$ using the \emph{truncated} preference lists.(If there are ties, fix a tie-breaking procedure and run deferred acceptance.) 

For each period-2 economy $\economyr^2=(\economyr^1,A_2,B_2)$, let $\matchingstar(\economyr^2)$ coincide with $\matchingstar(\economyr^1)$ for those agents who matched at $\economyr^1$ and have it coincide with the outcome of running deferred acceptance on $\cali(\matching^{\star^1},\economyr^2),\calj(\matching^{\star^1},\economyr^2)$, elsewhere.\footnote{Again, if there are indifferences, run deferred acceptance after having broken ties.}

Note that \matchingstar\ defined in this way is an element of $\Matchings_2$.

I note the following properties of $\matchingstar$. First, for all $\economyr^2$, $\matchingstar(\economyr^2)$ is stable for $\cali(\matchingstar_1,\economyr^1),\calj(\matchingstar_1,\economyr^2)$. Second, for each $\economyr^1$, \matchingstar\ satisfies the following: (i) there is no pair $(\agenti,\agentj)$ such that $\matchingstar(\economyr^1)(\agenti)\neq \agenti,\matchingstar(\economyr^1)(\agentj)\neq\agentj$ that prefer matching with each other than to match according to $\matchingstar(\economyr^1)$, i.e., $\matchingstar(\economyr^1)$ is stable amongst those who match in period $1$, and (ii) there is no $k\in\arrivali_1\cup\arrivalj_1$ such that $\matchingstar(\economyr^1)(k)\neq k$ and for all $\matchingb\in M_{\ds}(k,\matchingstar,\economyr^1)$, $k$ prefers $\matchingb$ to $\matchingstar$. To see that (ii) holds note the following. By construction $\matchingb^k$ can always be taken to be an element of $M_{\ds}(k,\matchingstar,\economyr^1)$ by having it coincide with $\matchingstar$ at all realizations $\economyrt$, $t\in\{1,2\}$ that do not (weakly) follow $\economyr^1$. Moreover, if $k$ is matched at $\economyr^1$, then it must be that $k$'s matching partner is (weakly) preferred to $\matchingb^k$.

Thus, to show that $\matchingstar\in\ds_2(\Gtwo)$, it remains to argue that it is pairwise stable. To do so, it only remains to check that there is no realization $\economyr^1=(\arrivali_1,\typea_1,\arrivalj_1,\typeb_1)$ and pair $(\agenti,\agentj)\in\arrivali_1\times\arrivalj_1$ who prefer to match together at $\economyr^1$ rather than match according to $\matchingstar$. Toward a contradiction, suppose there is such a pair. Since $\matchingstar$ is stable amongst those who match in period $1$, it cannot be the case that both \agenti\ and \agentj\ are matched at $\economyr^1$. Then, either \agenti\ or \agentj\ (or both) declared the other unacceptable at $\economyr^1$ and is unmatched at $\economyr^1$. Without loss of generality, assume that it is \agentj\, then
\begin{align}\label{eq:aunacceptabletob}
\genericvjone<V(\agentj,\matching_{\economyr^1}^b,\economyr^1).
\end{align}
On the other hand, it must be that $\genericvjone>V(\agentj,\matchingstar,\economyr^1)$, since (\agenti,\ \agentj)\ form a pairwise block of $\matchingstar$. It then follows that
\begin{align}\label{eq:wronglowerbound}
V(\agentj,\matchingstar,\economyr^1)<V(\agentj,\matchingb_{\economyr^1}^\agentj,\economyr^1).
\end{align}
Note, however, that \matchingstar\ satisfies the following: (i) $\matchingstar(\economyr^1)(\agentj)=\agentj$, (ii) $\matchingstar(\economyr^1,\cdot)\in\ds_1(\Gtwo(\matchingstar_1,\economyr^1))$, and (iii) $\matchingstar(\economyr^1)$ is stable amongst those who match in period $1$. Then, $\matchingstar\in M_{\ds}(\agentj,\matchingstar,\economyr^1)$. This, together with \autoref{eq:wronglowerbound} contradicts the definition of $\matchingb_{\economyr^1}^\agentj$. By assuming instead that $\agenti$ is unmatched, a similar contradiction follows. I conclude that \matchingstar\ is pairwise stable.

This concludes the proof that $\mathbf{P(2)}=1$.

I now prove the inductive step. Assume now that $\mathbf{P(\terminal^\prime)}=1$ for all $\terminal^\prime<\terminal$. I show that $\mathbf{P(\terminal)}=1$. 

Given a realization $\economyrterminal=(\arrivali_1,\typea_1,\arrivalj_1,\typeb_1,\dots,\arrivali_\terminal,\typea_\terminal,\arrivalj_\terminal,\typeb_\terminal)$ in the support of $\Gterminal$, consider the following procedure to construct $\matchingstar(\economyrt)$ for each truncation $\economyrt$ of $\economyrterminal$. For each $1\leq t\leq\terminal-1$, having defined $(\matchingstar(\economyr^1),\dots,\matchingstar(\economyr^{t-1}))$, let $\langle\cali(\matchingstaruptot,\economyrt),\calj(\matchingstaruptot,\economyrt)\rangle$ denote the agents available to match in period $t$ (if \periodone\, then $\cali(\emptyset,\economyr^1)=A_1,\calj(\emptyset,\economyr^1)=B_1$.) For each $k\in\cali(\matchingstaruptot,\economyrt)\cup\calj(\matchingstaruptot,\economyrt)$, let $\matchingb_{\economyrt}^k$ denote the matching $\matching$ that coincides with $\matchingstar$ on $(\economyr^1,\dots,\economyr^{t-1})$ and minimizes $k$'s payoff at $\economyrt$ amongst all the matchings that also satisfy\footnote{That is, if $k=\agenti$, I am minimizing $U(\agenti,\cdot, \economyrt)$, while if $k=\agentj$, I am minimizing $V(\agentj,\cdot,\economyrt)$.} 
\begin{enumerate}
\item $\matching(\economyrt)(k)=k$
\item $(\matching(\economyr^{t+\dateindex}))_{s\in\{1,\dots,\terminal-t\}}\in\ds_{T-t}(G_{\terminal-t}(\matching^t,\economyrt))$.
\item $\matching(\economyrt)$ is stable amongst those who match in period $t$.
\end{enumerate}
By the inductive hypothesis, $\ds_\dateindex$ is non-empty valued for $\dateindex<\terminal$. Since $t\geq 1$, then $\ds_{\terminal-t}(\cdot)$ is non-empty. Thus, for all $k\in\cali(\matchingstaruptot,\economyrt)\cup\calj(\matchingstaruptot,\economyrt)$, $\matchingb_{\economyrt}^k$ is well-defined since the above set of matchings is non-empty.

For each $\agenti\in\cali(\matchingstaruptot,\economyrt)$, consider the preference list on $\calj(\matchingstaruptot,\economyrt)$ defined as follows: $\agentj\in \calj(\matchingstaruptot,\economyrt)$ is acceptable to $\agenti$ only if $\genericuit\geq U(\agenti,\matchingb_{\economyrt}^\agenti,\economyrt)$ and is not acceptable to $\agenti$ otherwise. Moreover, if $\agentj,\agentjp\in\calj(\matchingstaruptot,\economyrt)$ are acceptable to $\agenti$, then order them according to $u(\typeauptot(i),\cdot)$ in $\agenti$'s preference list.

Similarly, for $\agentj\in\calj(\matchingstaruptot,\economyrt)$, consider the preference list on $\cali(\matchingstaruptot,\economyrt)$ defined as follows: $\agenti\in\cali(\matchingstaruptot,\economyrt)$ is acceptable to $\agentj$ only if $\genericvjt\geq V(\agentj,\matchingb_{\economyrt}^\agentj,\economyrt)$ and is not acceptable to $\agentj$ otherwise. Moreover, if $\agenti,\agentip\in\cali(\matchingstaruptot,\economyrt)$ are acceptable to $\agentj$, then order them according to $v(\cdot,\typebuptot(\agentj))$ in $\agentj$'s preference list.

Define $\matchingstar(\economyrt)$ to be the outcome of deferred acceptance on $\langle\cali(\matchingstaruptot,\economyrt),\calj(\matchingstaruptot,\economyrt)\rangle$ using the above preference lists.

For $t=\terminal$, let $\matchingstar(\economyrterminal)$ coincide with $\matchingstar(\economyr^{\terminal-1})$ for all those who have matched at $\economyr^{\terminal-1}$ and 
with the outcome of running deferred acceptance with agents on $\cali(\matchingstaruptoterminal,\economyrterminal)$ proposing to agents on $\calj(\matchingstaruptoterminal,\economyrterminal)$, elsewhere.

Note that this procedure on all $\economyrt$ and all $1\leq t\leq\terminal$ implies that \matchingstar\ is an element of $\Matchings_\terminal$. I now make some observations about $\matchingstar$.

First, for all periods $t\in\{1,\dots,\terminal\}$, and all realizations \economyrt,\ $\matchingstar(\economyrt)$ is stable amongst those who match at \economyrt.\ Second, for all periods $t\in\{1,\dots,\terminal-1\}$, and all realizations \economyrt,\ there is no $k\in\cali(\matchingstaruptot,\economyrt)\cup\calj(\matchingstaruptot,\economyrt)$ such that $\matchingstar(\economyrt)(k)\neq k$ and for all $\matchingb\in M_{\ds}(k,\matchingstar,\economyrt)$, $k$ prefers $\matchingb$ to $\matchingstar$. To see that this holds note the following. By construction $\matchingb_{\economyrt}^k$ can always be taken to be an element of $M_{\ds}(k,\matchingstar,\economyrt)$ by having it coincide with $\matchingstar$ at all realizations $\economyr^\dateindex$ that do not follow \economyrt.\ Moreover, if $k$ is matched at $\economyrt$, it must be that their matching partner is (weakly) preferred to $\matchingb_{\economyrt}^k$.

Fix $t=\terminal$ and note that for all $\economyrterminal$, $\matchingstar(\economyrterminal)$ is pairwise stable and individually rational. Let $t\leq\terminal-1$ and let $\economyrt$ denote a realization. Suppose that one has shown that $(\matchingstar(\economyr^{t+\dateindex}))_{\dateindex\in\{1,\dots,\terminal-t\}}$ is pairwise stable starting from period $t+1$ onward. Note that, in particular, this implies that $(\matchingstar(\economyr^{t+\dateindex}))_{\dateindex\in\{1,\dots,\terminal-t\}}\in\ds_{T-t}(G_{\terminal-t}(\matching^{\star^t},\economyrt))$. I now show that there are no pairwise blocks involving agents that have arrived by \economyrt.\ Toward a contradiction, suppose there exists a pair $(\agenti,\agentj)\in\cali(\matchingstaruptot,\economyrt)\times\calj(\matchingstaruptot,\economyrt)$ such that
\begin{align*}
\genericuit>U(\agenti,\matchingstar,\economyrt), \genericvjt>V(\agentj,\matchingstar,\economyrt).
\end{align*}
Since $\matchingstar(\economyrt)$ is stable amongst those who match in period $t$, then it must be that either $\agenti$ or $\agentj$ are unmatched at $\matchingstar(\economyrt)$. Without loss of generality supposed that it is \agentj.\ Thus, \agentj\ had \agenti\ as unacceptable in the truncated preference lists  so that:
\begin{align}
\genericvjt<V(\agentj,\matchingb_{\economyrt}^\agentj,\economyrt).
\end{align}
Putting the above inequalities together, it follows that
\begin{align}\label{eq:wronglowerboundt}
V(\agentj,\matchingstar,\economyrt)<V(\agentj,\matchingb_{\economyrt}^\agentj,\economyrt).
\end{align}
Note, however, that $\matchingstar\in M_{\ds}(\agentj,\matchingstar,\economyrt)$.\footnote{This, of course, uses that we have shown that $\matchingstar(\economyrt,\cdot)$ is pairwise stable and  has no blocks by waiting.} This, together with \autoref{eq:wronglowerboundt}, contradicts the definition of $\matchingb_{\economyrt}^\agentj$. Thus, $\matchingstar$ is pairwise stable from period $t$ onward. Proceeding this way, one concludes that $\matchingstar$ is pairwise stable. 
This proves $\mathbf{P(\terminal)}=1$.
\end{document}